\documentclass[11pt]{article}
\usepackage[utf8]{inputenc}
\usepackage[T1]{fontenc}
\usepackage[top=1in, bottom=1in, left=1in, right=1in, letterpaper]{geometry}

\usepackage{dsfont}
\usepackage{amsmath}
\usepackage{amsthm}
\usepackage{amssymb}
\usepackage{graphicx}
\usepackage{url}
\usepackage{nicefrac}
\usepackage{enumitem}
\usepackage{tikz} 
\usepackage{xspace}
\usepackage[noend]{algpseudocode}
\usepackage{comment}
\usepackage{hyperref}
\usepackage{subcaption}
\usepackage{authblk}

\usepackage{thmtools} 
\usepackage{thm-restate}

\usepackage{color}

\definecolor{blueblack}{rgb}{0,0,.7}

\newcounter{sideremark}

\definecolor{Darkblue}{rgb}{0,0,0.4}
\definecolor{Brown}{cmyk}{0,0.61,1.,0.60}
\definecolor{Purple}{cmyk}{0.45,0.86,0,0}
\definecolor{brickred}{rgb}{0.8, 0.25, 0.33}

\theoremstyle{plain}
\newtheorem{theorem}{Theorem}
\newtheorem{invariant}{Invariant}
\newtheorem{lemma}{Lemma}
\newtheorem{Claim}{Claim}
\newtheorem{definition}{Definition}

\newtheorem{remark}{Remark}
\newtheorem{corollary}{Corollary}

\newcommand{\eps}{\varepsilon}

\newcommand{\calP}{\mathcal{P}}

\newcommand{\OPT}{\mathsf{OPT}}

\newcommand{\APX}{\mathsf{APX}}

\usepackage[colorinlistoftodos,textsize=tiny,textwidth=2cm,color=red!25!white,obeyFinal]{todonotes}

\ifdefined\DEBUG
\newcommand{\fab}[1]{\textcolor{red}{#1}}
\newcommand{\mi}[1]{\textcolor{blue}{#1}}
\newcommand{\af}[1]{\textcolor{orange}{#1}}

\newcommand{\mig}[1]{\todo[color=cyan!100!black!50]{M: #1}}
\newcommand{\afr}[1]{\todo[color=green!100!black!50]{A: #1}}
\newcommand{\fabr}[1]{\todo[color=red!100!black!50]{F: #1}}
\else
\newcommand{\fab}[1]{#1}
\newcommand{\mi}[1]{#1}
\newcommand{\af}[1]{#1}

\newcommand{\fabr}[1]{}
\newcommand{\afr}[1]{}
\newcommand{\mig}[1]{}
\fi
    
\date{}

\title{A PTAS for Triangle-Free 2-Matching\footnote{The first 2 authors are partially supported by the SNF Grant 200021\_200731 / 1.}}

\author[1]{Miguel {Bosch-Calvo}\thanks{miguel.boschcalvo@idsia.ch}}
\author[1]{\mbox{Fabrizio Grandoni}\thanks{fabrizio.grandoni@idsia.ch}}
\author[2]{Afrouz Jabal Ameli\thanks{a.jabal.ameli@tue.nl}}
\affil[1]{IDSIA, USI-SUPSI, Switzerland}
\affil[2] {TU Eindhoven, Netherlands}

\begin{document}

\maketitle             

\begin{abstract}
\noindent In the Triangle-Free (Simple) 2-Matching problem we are given an undirected graph $G=(V,E)$. Our goal is to compute a maximum-cardinality $M\subseteq E$ satisfying the following properties: (1) at most two edges of $M$ are incident on each node (i.e., $M$ is a 2-matching) and (2) $M$ does not induce any triangle. In his Ph.D. thesis from 1984,  Harvitgsen presents a complex polynomial-time algorithm for this problem, with a very complex analysis. This result was never published in a journal nor reproved in a different way, to the best of our knowledge. 

In this paper we have a fresh look at this problem and present a simple PTAS for it based on local search. Our PTAS exploits the fact that, as long as the current solution is far enough from the optimum, there exists a short augmenting trail (similar to the maximum matching case).        
\end{abstract}

\section{Introduction}

Let $G=(V,E)$ be an undirected graph. A subset of edges $M\subseteq E$ is a (simple\footnote{We remark that here we are focusing on \emph{simple} 2-matching, meaning that we are allowed to include in $M$ at most one copy of each edge of $G$: therefore we will omit ``simple''.}) 2-matching if, for every node $v\in V$, at most two edges of $M$ are incident to $v$. In particular, $M$ is a collection of node-disjoint paths and cycles. 
It is easy to compute a 2-matching of maximum cardinality via a reduction to maximum matching.

In this paper we focus on the problem of computing a maximum-cardinality triangle-free 2-matching in $G$, i.e. a 2-matching $M$ which has the extra property not to induce any triangle (i.e., $(V,M)$ does not contain triangles). This extra property makes the problem substantially harder, in particular the above reduction to maximum matching does not work. In his Ph.D. thesis from 1984, Harvitgsen \cite{HartvigsenPhdThesis} presents a complex polynomial-time algorithm for this problem, with a very complex analysis. To the best of our knowledge, this result was never published in a journal, and no alternative or independent proof of it is known in the literature. 

Motivated by the above status of the literature, in this paper we have a fresh look at \fab{this} problem, and we present a drastically different and substantially simpler (though yet not elementary) proof of a slightly weaker result, namely a PTAS for the problem. 
\begin{theorem}\label{thr:main}
There exists a PTAS for triangle-free 2-matching.    
\end{theorem}
An overview of our approach is given in Section \ref{sec:overview}.

\subsection{Related Work}

The $C_k$-Free 2-matching problem is the problem of finding a 2-matching not inducing any $C_k$, i.e. a cycle of length $k$. In particular, the triangle-free 2-matching problem is the special case of the above problem with $k=3$. Hartvigsen~\cite{HARTVIGSEN2006} showed that the square-free case, i.e. $k=4$, in bipartite graphs is polynomial-time solvable (see also Pap~\cite{Pap2007} and Babenko~\cite{Babenko12}).
B{\'e}rczi and Kobayashi~\cite{BERCZI2012565} showed that square-free $2$-matching can be solved in polynomial time if the input graph is sub-cubic. 

Another interesting variant of the problem is the maximum $C_{\le k}$-free 2-matching problem, where the objective is to find a maximum-cardinality $2$-matching not containing any cycle of length at most $k$. Papadimitriou showed that the case $k=5$ is NP-hard (his proof is described by Cornu{\'{e}}jols and Pulleyblank in 
\cite{CP80}). The complexity of the case $k=4$ is open, while the case $k=3$ is equivalent to the triangle-free 2-matching problem.

One can consider the natural weighted variant of these problems, where the input graph has non-negative edge-weights and one wishes to find a solution of maximum total weight. 
Hartvigsen and Li~\cite{HL07} presented the first polynomial-time algorithm for weighted triangle-free $2$-matching in sub-cubic graphs (see also~\cite{KOBAYASHI2010197}).
Finding a 2-matching of maximum total weight that has no cycles of length at most $4$ is NP-hard (\cite{Kiraly09,BERCZI2012565}).

The problem gets easier if we are allowed to take multiple copies of the same edge (i.e., the 2-matching is not simple). In particular, the weighted triangle-free non-simple 2-matching problem can be solved exactly in polynomial time \cite{CP80}.

One of the reasons of interest of triangle-free 2-matchings is that they can be used as a subroutine to derive improved approximations algorithms for the 2-Edge-Connected Spanning Subgraph problem (2ECSS). We recall that in 2ECSS we are given an undirected graph, and our goals is to compute a subgraph with the minimum number of edges which is 2-edge-connected. 2ECSS is APX-hard \cite{CL99,F98}, and a lot of work was devoted to the design of approximation algorithms for it \cite{CSS01,HVV19,KV94,SV14}, culminating with a recent $1.326$ approximation \cite{GGJ23}. The latter paper exploits the fact that a minimum-cardinality 2-edge cover (which can be computed in polynomial time) provides a lower bound on the optimum solution. Very recently Kobayashi and Noguchi \cite{kobayashi2023approximation} observed that a maximum-cardinality triangle-free 2-matching can be used to build a triangle-free 2-edge-cover whose size is an alternative lower bound to the optimal solution to 2ECSS (a similar idea appeared earlier in \cite{Jothi03}). Combining this observation with the analysis in 
\cite{GGJ23} and the polynomial-time algorithm for triangle-free 2-matching claimed by Hartvigsen \cite{HartvigsenPhdThesis}, one obtains an improved $1.3+\eps$ approximation for 2ECSS. The same result now can be achieved also exploiting the PTAS presented in this paper. \fab{In our opinion, triangle-free 2-matchings might be helpful also for the vertex-connectivity variant 2VCSS of 2ECSS \cite{BCGJ23,GVS93,HV17,KV94}.} 

\fab{2ECSS and 2VCSS fall in the area of Survivable Network Design, an area that recently attracted a lot of attention in terms of approximation algorithms (see, e.g., \cite{BGJ23,FGKS18,GJT22,N20,TZ23stoc}).}

\section{Preliminaries and Notation}

We use standard graph theoretical notation. Let $G=(V, E)$ be an undirected graph. 
Given a node $u\in V$, we denote the neighborhood of $u$ in $G$ as $N_G(u)$ or $N(u)$ when $G$ is clear from the context. Also, given a subset of edges $S\subseteq E$, we define as $N_S(u)$ the neighborhood of $u$ in the graph $G'=(V, S)$. A triangle involving the nodes $u$, $v$ and $w$ is simply denoted by $uvw$. 

A trail $P$ (of length $|P|=k-1$) is a sequence of distinct consecutive edges $u_1u_2,u_2u_3,\ldots,u_{k-1}u_k$. The nodes $u_1,\ldots,u_k$ are not necessarily distinct. We say that $u_1$ and $u_k$ (where possibly $u_1=u_k$) are \fab{the first and the last node of the trail, resp.}, and they are both endpoints of the trail. 
Given two trails $P_1=u_1u_2,\ldots,u_{k-1}u_k$ and $P_2=u_kv_2,\ldots,v_{h-1}v_h$ involving distinct edges, $P_1\circ P_2=u_1u_2,\ldots,u_{k-1}u_k,u_kv_2,\ldots,v_{h-1}v_h$ denotes their concatenation. 
Given a trail $P$, we will also use $P$ to denote the corresponding set of edges $\{u_1u_2,\ldots,u_{k-1}u_k\}$ (neglecting their order): \fab{the meaning will be clear from the context.}

\paragraph{Notation in the figures.} We will refer to an optimal solution $\OPT$ and an approximate solution $\APX$. In the coming figures, the edges colored in red belong to $\APX\setminus\OPT$, while the edges colored in blue to $\OPT\setminus\APX$. The edges in $\APX\cap \OPT$ are colored with red and blue. We will construct a collection of edge-disjoint alternating trails \fab{(as defined in next section)} $P_1,\ldots,P_k$. The edges in $P_1,...,P_{k-1}$ are dashed, while the edges in $P_k$ have an arrow pointing towards the end of the trail (we use a double arrow when the direction is not relevant).
Black and gray edges represent edges in $\APX\cup\OPT$ and potentially in $\APX\cup\OPT$, resp. (and that potentially could satisfy a subset of the above conditions). 

\section{Overview of Our Approach}
\label{sec:overview}

\begin{figure}[t]
    \centering
    \caption{PTAS for triangle-free $2$-matching. Here $\eps\in (0,1]$ is an input parameter.}
    \begin{algorithmic}[1]
        \State $\APX\gets \emptyset$
        \While{there exists an augmenting trail $P$ for $\APX$ containing at most $2/\eps$ edges} 
                    \State $\APX\gets \APX\triangle P$.
        \EndWhile
        \State \textbf{return} $\APX$
    \end{algorithmic}
    \label{alg:ptas}
\end{figure}
Our PTAS is simply based on local search (see also  Figure \ref{alg:ptas}). We need to define certain quantities that mimic the notion of alternating path and augmenting path in matching theory. Given a triangle-free 2-matching $\APX$, we say that a trail $P$ is alternating w.r.t. $\APX$ if $P$ does not contain consecutive edges in $\APX$ nor in $E\setminus \APX$, i.e. the edges of $\APX$ and $E\setminus \APX$ alternate in $P$. An alternating trail $P$ is augmenting for $\APX$ if $\APX \triangle P$\footnote{We recall that,
given any two sets $A$ and $B$, their symmetric difference  $A\triangle B:=(A\setminus B) \cup (B\setminus A)$ is the subset of elements contained in exactly one of the two sets $A$ and $B$.} is a triangle-free 2-matching of size $|\APX|+1$. An example of augmenting trail is given in Figure \ref{fig:exampleTrail}.

\begin{figure}[ht]
    \centering

\includegraphics{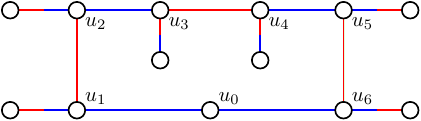}
    \caption{An augmenting trail is given by $u_0u_1,u_1u_2,u_2u_3,u_3u_4,u_4u_5,u_5u_6,u_6u_0$. Notice that no augmenting path exists in the matching sense.}     
    \label{fig:exampleTrail}
\end{figure}


Our PTAS simply increases the size of the current solution $\APX$ (initially empty) by exploiting any augmenting trail of length up to $2/\eps$, halting when no such trail exists. It is obvious that this algorithm runs in polynomial time \fab{for any constant $\eps>0$}. In order to analyze its approximation factor, we will prove the following non-trivial result.
\begin{lemma}\label{lem:augmentingTrails}
Let $\APX$ be a triangle-free 2-matching for $G=(V,E)$ and $\OPT$ be a maximum-cardinality triangle-free 2-matching for $G$ which maximizes $|\OPT\cap \APX|$. The\fab{n} there exist at least $|\OPT|-|\APX|$ edge disjoint augmenting trails for $\APX$ contained in $\APX\triangle \OPT$.    
\end{lemma}
We remark that the above lemma is very similar to an analogous lemma for standard matching (with edge-disjoint augmenting trails replaces by node-disjoint augmenting paths). However, a critical difference is that augmenting simultaneously along multiple augmenting trails \fab{obtained via Lemma \ref{lem:augmentingTrails}} might lead to an infeasible solution. In other words, it might happen that, given two edge disjoint augmenting trails $P_1$ and $P_2$ for $\APX$, $\APX\triangle (P_1\cup P_2)$ is not a triangle-free 2-matching.

Given the above lemma, it is easy to prove Theorem \ref{thr:main}.
\begin{proof}[Proof of Theorem \ref{thr:main}]
Consider the algorithm from Figure \ref{alg:ptas}. Searching for a valid augmenting trail as in the while loop can be implemented in time $n^{O(2/\eps)}$, where $n=|V|$. Each time one such trail is identified, $|\APX|$ grows by $1$. Hence this step cannot be repeated more than $n$ times (otherwise some node would have degree more than $2$ w.r.t. $\APX$). 

Consider next the approximation factor. Let $\APX$ be the returned solution (which is obviously feasible by construction). Let $\OPT$ be an optimal solution maximizing $|\APX\cap \OPT|$. We claim that $|\APX|\geq (1-\eps)|\OPT|$. Assume by contradiction that this is not the case. By applying Lemma \ref{lem:augmentingTrails}, there exist $|\OPT|-|\APX|$ edge-disjoint augmenting trails for $\APX$ in $\APX\triangle \OPT$. Thus there exists an augmenting trail $P$ of length
$
|P|\leq \frac{|\APX\triangle \OPT|}{|\OPT|-|\APX|}\leq \frac{2|\OPT|}{\eps|\OPT|}=\frac{2}{\eps}.
$
This $P$ satisfies the condition of the while loop, contradicting the fact that $\APX$ is the solution returned by the algorithm.
\end{proof}
Another simple corollary of Lemma \ref{lem:augmentingTrails} is that a \emph{maximal} triangle-free 2-matching is a 2-approximate solution, again something analogous to standard matchings.
\begin{lemma}\label{lem:maximalApproximation}
A maximal triangle-free 2-matching (under the insertion of one edge) is a 2-approximation for the triangle-free 2-matching problem. 
\end{lemma}
\begin{proof}
Let $\APX$ be a maximal triangle-free 2-matching and $\OPT$ be a maximum triangle-free 2-matching that maximizes $|\APX\cap \OPT|$. Assume by contradiction that $|\APX|<|\OPT|/2$. By Lemma \ref{lem:augmentingTrails}, there exist $|\OPT|-|\APX|>|\OPT|/2$ edge-disjoint augmenting trails in $\APX\triangle \OPT$. Thus there exists one augmenting trail $P$ with  
$
|P|<\frac{|\APX\triangle \OPT|}{|\OPT|/2}\leq \frac{3|\OPT|/2}{|\OPT|/2}=3.
$
In particular, $|P|\leq 2$. Since an augmenting trail contains an odd number of edges, $|P|=1$, which contradicts the maximality of $\APX$.
\end{proof}

Let us sketch the proof of Lemma \ref{lem:augmentingTrails}. Recall that our goal is to build a collection $P_1,\ldots,P_q$ of $q:=|\OPT|-|\APX|$ edge-disjoint augmenting trails for $\APX$ contained in $\APX\triangle \OPT$. Our construction is iterative: given the augmenting trails $P_1,\ldots,P_{k-1}$, $k\leq q$, we will construct $P_k$.

The construction of $P_k$ itself is iterative. Let $\calP_{j}:=\cup^{j}_{i=1}P_i\subseteq \APX\triangle\OPT$. We say that the edges of $\calP_{k-1}$ are \emph{used}, and the remaining edges of $(\APX\triangle \OPT)\setminus \calP_{k-1}$ are \emph{free}. The path $P_k$ is obtained by concatenating \emph{chunks} \fab{of free edges}, where a chunk is defined as follows (see also Figure \ref{fig:possibleChunks}).
\begin{definition}\label{def:chunk}
    A \emph{chunk} $c$ is an alternating trail $u_1u_2,\ldots,u_{k-1}u_k$, $2\leq k \leq 4$, with $c\subseteq \APX\triangle \OPT$ and $u_iu_{i+1}u_{i+2}$ is a triangle in $\APX\cup\OPT$ for all $i\in [1, k-2]$.
\end{definition}
\begin{figure}[t]
    \centering
    \includegraphics{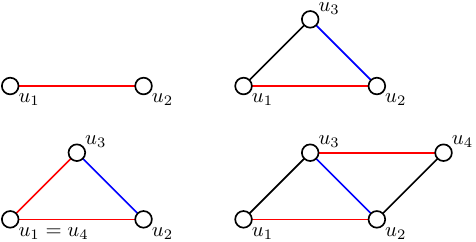}
    \caption{All possible chunks starting with an edge of $\APX\setminus\OPT$, the other case (\fab{i.e.,} the starting edge is in $\OPT\setminus\APX$) being symmetrical.} 
    \label{fig:possibleChunks}
\end{figure}
We keep the invariant that $P_k$ is an alternating trail starting with a free edge in $\OPT\setminus \APX$, and that $\APX\triangle P_k$ does not contain any triangle. The construction will guarantee that, while $P_k$ is not augmenting itself, it is possible to append another chunk to $P_k$.

For the sake of our analysis, we will need to build the collection of trails $P_1,\ldots,P_k$ in a carefully restricted way. First we identify some nodes as deficient.
\begin{definition}
    Given a subset of edges $E'\subseteq E$, we say that a node $u\in V$ is deficient w.r.t. $E'$ if $|N_{\APX\setminus E'}(u)| < |N_{\OPT\setminus E'}(u)|$. If $E'=\emptyset$, then we simply say that $u$ is deficient. 
\end{definition}

Our collection of paths will satisfy the following invariant at any point of time.
\begin{invariant}\label{inv:main}
Let $P_1,\ldots,P_k$, $1\leq k\leq |\OPT|-|\APX|$, be the \fab{current collection of alternating trails at any point of our construction.}\fabr{I found previous version a bit confusing, also this one is not great}
Then for each $j=1,\ldots,k$ one has:
    \begin{enumerate}\itemsep0pt 
        \item $P_j=c^j_{1}\circ c^j_{2}\circ\dots\circ c^j_{l_j}$, where $c^j_{1},\dots, c^j_{l_j}$ are edge-disjoint chunks and $l_j\geq 1$.\label{prp:chunks}
        \item $P_j\subseteq (\APX\triangle \OPT)\setminus \calP_{j-1}$.\label{prp:disjoint}
        \item If $j<k$, $P_j$ is augmenting and the first and last nodes of $P_j$ are deficient w.r.t. $\calP_{j-1}$. Otherwise (i.e., $j=k$), the first node of $P_j$ is deficient w.r.t. $\calP_{j-1}$ and the first edge of $P_j$ is in $\OPT\setminus \APX$.\label{prp:startingEnding}
        \item ($\APX$-triangle) $\APX\triangle P_j$ contains no triangle.\label{prp:redTriangles}
        \item ($\OPT$-triangle) $\OPT\triangle \calP_j$ contains no triangle.\label{prp:blueTriangles}
    \end{enumerate}
\end{invariant}
Now let us motivate the different properties. Property \ref{prp:chunks} simply states that the trails $P_j$ are obtained by appending chunks one after the other. This property will be trivially maintained by construction. Property \ref{prp:disjoint} states that each $P_j$ is edge-disjoint from the previous trails and contains only edges in $\APX\triangle \OPT$. Also this property will be trivially maintained by construction. The intuition behind Property is \ref{prp:startingEnding} is as follows. Our goal is to build a collection of augmenting trails, hence it makes sense to assume that the first $k-1$ trails satisfy this property, while $P_k$ is the initial part of one augmenting trail (hence its first edge must be in $\OPT$). For technical reasons, that will be clearer from the proofs, it is convenient to assume that each start and end node of each $P_j$ is deficient  w.r.t. $\calP_{k-1}$.
\begin{remark}
The request that the last node \fab{of $P_j$, $j<k$,} is deficient has a counter-intuitive effect: it might happen that at some point $P_k$ is already an augmenting trail, and nonetheless we have to extend it with an extra chunk since its last node is not deficient w.r.t. $\calP_{k-1}$.
\end{remark}
The $\APX$-triangle Property \ref{prp:redTriangles} is fairly natural: our goal is to build a collection of augmenting trails, hence $\APX\triangle P_j$ cannot contain triangles. Notice that for $P_k$ we enforce this property even when it is not yet an augmenting trail. The $\OPT$-triangle Property \ref{prp:blueTriangles} is due to technical reasons: though we do not care of creating triangles in $\OPT\triangle \calP_k$, this turns out to be a very helpful property in our construction. Figure \ref{fig:optTriProperty} contains an example motivating the helpfulness of this property.
\begin{figure}[t]
    \centering
    \includegraphics{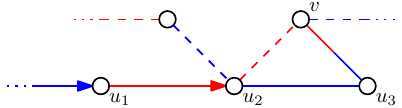}
    \caption{$u_2$ is the last node of $P_k$. The only way to keep extending $P_k$ is to add the trail $u_2u_3$. But then $u_2u_3v$ is a triangle in $\APX\triangle P_k$ and we cannot extend $P_k$ further, so this would be a dead end. If Invariant~\ref{inv:main}.\ref{prp:blueTriangles} was satisfied, this example would not exist since $u_2u_3v$ is a triangle in $\OPT\triangle\calP_k$.}
    \label{fig:optTriProperty}
\end{figure}

In order to make progress at each iteration, we exploit the following two lemmas. The first lemma (proof in Section \ref{sec:firstChunk}) is applied when the trail $P_k$ is augmenting and its last node is deficient w.r.t. $\calP_{k-1}$ (assuming $k<|\OPT|-|\APX|$, otherwise we are done). In this case we create a new alternating trail $P_{k+1}$ consisting of a single chunk and we add it to the current collection of alternating trails. Indeed, we apply this lemma also at the beginning of the algorithm to get the very first alternating trail $P_1$ in our construction.
\begin{lemma}\label{lem:firstChunk}
Let $P_1,\ldots,P_k$ be the current collection of alternating trails and $k<|\OPT|-|\APX|$. Suppose that $k=0$ (i.e., the collection is empty) or $P_k$ is augmenting and its last node is deficient w.r.t. $\calP_{k-1}$. Then there exists a chunk $P_{k+1}$ such that $P_1,\ldots,P_{k+1}$ satisfy Invariant \ref{inv:main}.
\end{lemma}
The second lemma (proof in Section \ref{sec:nextChunk}) is applied in any other case, namely when $P_k$ is not augmenting or it is augmenting but its last node is not deficient w.r.t. $\calP_{k-1}$. In this case we append one extra chunk at the end of $P_k$.
\begin{lemma}\label{lem:nextChunk}
Let $P_1,\ldots,P_k$ be the current collection of alternating trails, $k\leq |\OPT|-|\APX|$. Suppose that $P_k$ is not augmenting or it is augmenting but its last node is not deficient w.r.t. $\calP_{k-1}$. Then there exists a chunk $c$ such that $P_1,\ldots,P_{k-1},P'_{k}$ satisfy Invariant \ref{inv:main}, where $P'_k:=P_k\circ c$.
\end{lemma}

When applying the above lemmas, it might happen that multiple chunks satisfy the claim. In that case, \fab{our analysis requires to break such ties in a careful way.} The details of the tie-breaking rule will be given in Section \ref{sec:nextChunk}.

The proof of Lemma \ref{lem:augmentingTrails} follows easily.
\begin{proof}[Proof of Lemma \ref{lem:augmentingTrails}]
Assume $|\OPT|>|\APX|$, otherwise the claim trivially holds. We initialize our collection of alternating trails with the chu\af{n}k $P_1$ implied by Lemma \ref{lem:firstChunk}. Given the current set of alternating trails $P_1,\ldots,P_k$, we get a new collection containing more edges by applying Lemma \ref{lem:firstChunk} or \ref{lem:nextChunk}, depending on the cases. The process ends when $k=|\OPT|-|\APX|$ and $P_k$ is augmenting. Clearly this condition must hold after a polynomial number of updates since in each update we increase $|\calP_{k}|$ by at least $1$ and the total number of edges that we can add is at most $|\APX\triangle \OPT|\leq 2n$.     
\end{proof}
In the following we will assume that $|\OPT|>|\APX|$, and we will focus on the edges of $\APX\cup \OPT$, neglecting any other edge of $G$. Notice that, under this assumption, the degree of any node is at most $4$.

\paragraph{\bf Comparison with \cite{HartvigsenPhdThesis}.} Harvitgsen\fabr{Added this paragraph, check. \af{Looks very nice}} also exploits the notion of augmenting trail (augmenting path in his work). His algorithm is based on the search of a single augmenting trail at each iteration. To find such a trail efficiently, he has to deal with a large number of special subgraphs which play a role similar to blossoms in standard matching theory. Notice that we do not need to deal with such structures. Harvitgsen's algorithm implies that a solution is not maximum iff it admits an augmenting trail: it is possible to prove the same in a much simpler way with our approach, via a massive simplification of Lemma \ref{lem:augmentingTrails} (omitted here). As far as we can see, Lemma \ref{lem:augmentingTrails} is not implied by Harvitgsen's work.

\section{Proof of Lemma \ref{lem:firstChunk}}
\label{sec:firstChunk}

In this section we prove Lemma  \ref{lem:firstChunk}. To that aim, we will need the following two technical lemmas. 
 It is obvious that the average degree of the nodes w.r.t. $\OPT$  is not smaller than w.r.t. $\APX$. The following useful technical lemma shows that the same holds for every single node. 
\begin{lemma}[Degree Domination]\label{lem:blueMajority}
    Let $\APX$ be a triangle-free $2$-matching of a graph $G=(V, E)$, and $\OPT$ be a maximum-cardinality triangle-free $2$-matching of $G$ which maximizes $|\APX\cap \OPT|$. Then, for every $u\in V$, $|N_{\APX}(u)|\leq |N_{\OPT}(u)|$.
\end{lemma}
\begin{proof}
    Assume to get a contradiction that there is a node $u\in V$ such that $|N_{\APX}(u)| > |N_{\OPT}(u)|$.\fabr{Use of v,v' was confusing, reprased a bit} 

\begin{Claim}\label{lem:blueMajority:claim1}
\fab{Consider any $uv\in \APX\setminus \OPT$.} Then $\OPT\cup\{uv\}$ contains exactly one triangle $uvw$.
\end{Claim}
\begin{proof}
    Assume by contradiction that $\OPT\cup\{uv\}$ does not contain a triangle. Notice that the degree in $\OPT\cup\{uv\}$ of every node is the same as their degree in $\OPT$, except for $u$ and $v$, for which it is increased by $1$. Since $|N_{\OPT}(u)|\leq 1$, $|N_{\OPT\cup\{uv\}}(u)|\leq 2$. If $|N_{\OPT}(v)|<2$, then $\OPT\cup\{uv\}$ is a triangle-free $2$-matching of higher cardinality than $\OPT$, a contradiction. On the other hand, if $|N_{\OPT}(v)| = 2$, then since $uv\in \APX\setminus \OPT$, there must be an edge $e\in \OPT\setminus \APX$ incident to $v$. Thus $\OPT\cup\{uv\}\setminus\{e\}$ is clearly a simple $2$-matching, and it contains no triangle (because $\OPT\cup\{uv\}$ also contains no triangle). But then $\OPT\cup\{uv\}\setminus\{e\}$ is a maximum triangle-free $2$-matching with $|\APX\cap(\OPT\cup\{uv\}\setminus\{e\})|>|\APX\cap \OPT|$, a contradiction.

    Therefore, $\OPT\cup\{uv\}$ contains a triangle. Since $\OPT$ does not contain any triangle, every triangle in $\OPT\cup\{uv\}$ must contain the edge $uv$. Let $uvw$ be a triangle in $\OPT\cup\{uv\}$. $|N_{\OPT}(u)|\leq 1$, so $N_{\OPT\cup\{uv\}}(u)=\{v, w\}$. Thus, it must be that $uvw$ is the only triangle in $\OPT\cup\{uv\}$.     
\end{proof}
\begin{Claim}\label{lem:blueMajority:claim2}
\fab{Let $uvw$ be any triangle as in Claim \ref{lem:blueMajority:claim1}. Then}
$vw\in \APX\cap \OPT$.
\end{Claim}
\begin{proof}
    Since $vw\in \OPT\cup\{uv\}$, it must be that $vw\in\OPT$. Assume by contradiction that $vw\in \OPT\setminus \APX$. \fab{Using the fact that $|N_{\OPT}(u)|=1$}, $\OPT':=\OPT\cup\{uv\}\setminus\{vw\}$ is a $2$-matching. Furthermore, since $uvw$ is the only triangle in $\OPT\cup\{uv\}$, $\OPT'$ is also triangle-free, contradicting the definition of $\OPT$ \fab{since} $|\APX\cap \OPT'|>|\APX\cap \OPT|$. 
\end{proof}

Since $|N_{\APX}(u)| > |N_{\OPT}(u)|$, there exists at least one edge $u\fab{v'}\in\APX\setminus\OPT$. By Claims~\ref{lem:blueMajority:claim1} and~\ref{lem:blueMajority:claim2}, $\OPT\cup\{u\fab{v'}\}$ contains exactly one triangle $u\fab{v'}\mi{w'}$, with $\fab{v'}\mi{w'}\in\OPT\cap\APX$. One of the following two cases must happen:

\vspace{2.5mm}\noindent\textbf{(1) $\mathbf{N_{\APX\setminus \OPT}(u)=\{\fab{v'}\}}$.} Since $|N_{\APX\setminus \OPT}(u)|=1$ and $|N_{\APX}(u)| > |N_{\OPT}(u)| = |\{\mi{w'}\}|=1$, it must be that $u\mi{w'}\in \APX\cap \OPT$. But then $u\fab{v'}\mi{w'}$ is a triangle in $\APX$, a contradiction.

\vspace{2.5mm}\noindent\textbf{(2) $\mathbf{N_{\APX\setminus \OPT}(u)=\{\fab{v'}, \fab{v''}\}}$.} By Claims~\ref{lem:blueMajority:claim1} and~\ref{lem:blueMajority:claim2} \fab{applied to $uv''$}, $\OPT\cup\{u\fab{v''}\}$ must contain exactly one triangle $u\fab{v''}\mi{w''}$, with $\fab{v''}\mi{w''}\in \APX\cap \OPT$. Since $N_{\OPT}(u)=\{\mi{w'}\}$, it must be that $\mi{w'}=\mi{w''}$. But then $N_{\OPT}(\mi{w'})=\{u, \fab{v'}, \fab{v''}\}$, a contradiction to the fact that $\OPT$ is a $2$-matching.
\end{proof}
\begin{remark}
Lemma \ref{lem:augmentingTrails} would hold by replacing $\OPT$ with any triangle-free 2-matching $\OPT'$  satisfying the degree domination property as in Lemma \ref{lem:blueMajority}. It is easy to infer it from our proofs.
\end{remark}

The next lemma is very helpful to reduce the number of cases in our case analysis. \fab{Recall that $\calP_j:=\cup_{i=1}^{j}P_i$.}
\begin{lemma}[Parity]\label{lem:parityLemma}
Let $P_1,\ldots,P_{k-1}$ be a collection of edge-disjoint augmenting trails \fab{in $\APX\triangle \OPT$}. Let $P_k$ be an alternating trail contained in $(\APX\triangle\OPT)\setminus \calP_{k-1}$, starting with an edge in $\OPT\setminus \APX$ and whose first node is deficient w.r.t. $\calP_{k-1}$. Then, for every $u\in V$ excluding possibly the last node of $P_k$, one has 
\begin{enumerate}\itemsep0pt
    \item $|N_{\APX\triangle P_k}(u)|\leq 2$;\label{lem:parityLemma:1}
    \item $|N_{\OPT\triangle \calP_k}(u)|\leq 2$.\label{lem:parityLemma:2}
\end{enumerate}
\end{lemma}   
\begin{proof}
We start by observing the following.
\begin{Claim}\label{clm:parity}
For each node $w$ one has:
\begin{enumerate}\itemsep0pt
    \item $|N_{\OPT\triangle \calP_{k-1}}(w)|\leq |N_{\OPT}(w)|$;\label{clm:parity:1}
    \item $|N_{\APX}(w)|-|N_{\APX\setminus \calP_{k-1}}(w)|\leq |N_{\OPT}(w)|-|N_{\OPT\setminus \calP_{k-1}}(w)|$.\label{clm:parity:2}
\end{enumerate}
\end{Claim}
\begin{proof}
Let $\calP^w_{k-1}$ be the edges of 
$\calP_{k-1}$ incident to $w$. Notice that the inequalities do not change if we replace $\calP_{k-1}$ with $\calP^w_{k-1}$. Since the trails $P_1,\ldots,P_{k-1}$ are alternating, we can pair each edge $aw$ of $\APX\cap \calP^w_{k-1}$ with an edge $wb$ of $\OPT\cap \calP^w_{k-1}$ with the exception of the first and last edges of each $P_1,\ldots,P_{k-1}$ (which belong to $\OPT$) in case they \fab{have $w$ as an endpoint}. Each such pair of edges contributes with $+1$ to  both $|N_{\OPT}(w)|-|N_{\OPT\setminus \calP^w_{k-1}}(w)|$ and $|N_{\APX}(w)|-|N_{\APX\setminus \calP^w_{k-1}}(w)|$, and with $0$ to $|N_{\OPT}(w)|-|N_{\OPT\triangle \calP^w_{k-1}}(w)|$. The remaining edges of $\calP^w_{k-1}$ contribute with $+1$ to both $|N_{\OPT}(w)|-|N_{\OPT\setminus \calP^w_{k-1}}(w)|$ and $|N_{\OPT}(w)|-|N_{\OPT\triangle \calP^w_{k-1}}(w)|$, and with $0$ to $|N_{\APX}(w)|-|N_{\APX\setminus \calP^w_{k-1}}(w)|$. The claim follows. 
\end{proof}

\begin{Claim}\label{claim:deficient}
    \fab{Let $w$ be deficient} w.r.t. $\calP_{k-1}$. Then $w$ is deficient.
\end{Claim}
\begin{proof}
    \fab{By assumption,} $|N_{\OPT\setminus \calP_{k-1}}(w)|-|N_{\APX\setminus \calP_{k-1}}(w)| > 0$. By Claim~\ref{clm:parity}.\ref{clm:parity:2}, $|N_{\OPT}(w)|-|N_{\APX}(w)|\geq |N_{\OPT\setminus \calP_{k-1}}(w)|-|N_{\APX\setminus \calP_{k-1}}(w)| > 0$, which implies that $w$ is deficient.
\end{proof}

    \vspace{2.5mm}\noindent\textbf{(1)} Consider all the edges $P^u_k$ of $P_k$ incident to $u$. Notice that $|N_{\APX\triangle P_k}(u)| = |N_{\APX\triangle P^u_k}(u)|$. For each instance of $u$ which appears along $P_k$ not as an endpoint, there is a pair of edges $vu,uw$ along $P_k$, one belonging to $\OPT$ and the other to $\APX$. These pairs of edges contribute $0$ to $|N_{\APX\triangle P^u_k}(u)|-|N_{\APX}(u)|$. Since $u$ is not the last node of $P_k$, any remaining edge in $P_k$ incident to $u$ must be the first edge $uv$ of $P_k$, hence in particular $uv\in \OPT$. This edge contributes with $+1$ to $|N_{\APX\triangle P^u_k}(u)|-|N_{\APX}(u)|$. If $u$ is not the first node of $P_k$, one has $|N_{\APX\triangle P^u_k}(u)| = |N_{\APX}(u)|\leq 2$. Otherwise (namely $u$ is the first node of $P_k$), one has $|N_{\APX\triangle P^u_k}(u)| = |N_{\APX}(u)|+1\leq 2$. In the second inequality we used the fact that, since $u$ is the first node of $P_k$, $u$ is deficient w.r.t. $\calP_{k-1}$, and thus by Claim~\ref{claim:deficient}, $u$ is deficient (hence $|N_{\APX}(u)|\leq 1$).

    \vspace{2.5mm}\noindent\textbf{(2)} By Claim \ref{clm:parity}.\ref{clm:parity:1} one has $|N_{\OPT\triangle \calP_{k-1}}(u)|\leq |N_{\OPT}(u)|\leq 2$. It is sufficient to show that $|N_{\OPT\triangle \calP_{k}}(u)|\leq |N_{\OPT\triangle \calP_{k-1}}(u)|$. For each instance of $u$ which appears along $P_k$ not as an endpoint, there is a pair of edges $vu,uw$ along $P_k$, one belonging to $\OPT$ and the other not. These pairs of edges contribute $0$ to $|N_{\OPT\triangle \calP_{k-1}}(u)|-|N_{\OPT\triangle \calP_{k}}(u)|$. Since $u$ is not the last node of $P_k$, any remaining edge in $P_k$ incident to $u$ must be the first edge $uv$ of $P_k$, hence in particular $uv\in \OPT$. This edge contributes with $+1$ to $|N_{\OPT\triangle \calP_{k-1}}(u)|-|N_{\OPT\triangle \calP_{k}}(u)|$. The claim follows. 
\end{proof}

\begin{corollary}[Unique Triangle]\label{cor:uniqueTriangle}
    Let $P_1,\ldots,P_{k-1}$ be a collection of edge-disjoint augmenting trails \fab{in $\APX\triangle \OPT$}. Let $P_k$ be an alternating trail contained in $(\APX\triangle\OPT)\setminus \calP_{k-1}$, starting with an edge in $\OPT\setminus \APX$ and whose first node is deficient w.r.t. $\calP_{k-1}$. Then every edge $uv\in P_k$ belongs to at most one triangle in $\APX\triangle P_k$ and to at most one triangle in $\OPT\triangle\calP_k$.
\end{corollary}
\begin{proof}
    For every edge $uv$ of $P_k$, at least one of the endpoints is not the last node of $P_k$, assume w.l.o.g. that $u$ is not the last node of $P_k$. By \fab{the Parity} Lemma \ref{lem:parityLemma}, $|N_{\APX\triangle P_k}(u)|\leq 2$. Since a triangle in $\APX\triangle P_k$ containing the node $u$ must contain two edges in $\APX\triangle P_k$ incident to $u$, this triangle must be unique. A similar argument holds for $\OPT\triangle\calP_k$.
\end{proof}

\begin{proof}[Proof of Lemma \ref{lem:firstChunk}]
    Since $k < |\OPT| - |\APX|$ and every trail $P_i, 1\leq i\leq k$, has exactly one more edge in $\OPT$ than in $\APX$, there must exist a node $u_0$ that is deficient w.r.t. $\calP_{k}$. By definition of deficient w.r.t. $\calP_{k}$, there exists an edge $u_0u_1\in \OPT\setminus (\APX\cup\calP_{k})$. Notice that $u_0u_1$ is a chunk. 

     If $P_{k+1}:=u_0u_1$ satisfies the claim we are done, hence assume this is not the case. Observe that $P_1,\ldots,P_{k+1}$ satisf\fab{y} Invariants~\ref{inv:main}.\ref{prp:chunks},~\ref{inv:main}.\ref{prp:disjoint} and~\ref{inv:main}.\ref{prp:startingEnding}. Since $u_0u_1\in \OPT\setminus \APX$ and $\OPT\triangle \calP_k$ is triangle-free, Invariant~\ref{inv:main}.\ref{prp:blueTriangles} ($\OPT$-triangle) is also satisfied. Thus Invariant~\ref{inv:main}.\ref{prp:redTriangles} ($\APX$-triangle) must be violated. In more detail, $\APX\cup \{u_0u_1\}$ must contain a triangle (and $u_0u_1$ must belong to that triangle since $\APX$ is triangle-free). By the Unique Triangle Corollary~\ref{cor:uniqueTriangle}, there is only one such triangle: let $u_0u_1u_2$ be the only triangle in $\APX\cup \{u_0u_1\}$. Notice that we must have $u_0u_2, u_1u_2\in \APX$. We distinguish two cases:

    \vspace{2.5mm}\noindent\textbf{(1) $\mathbf{u_0u_2\in \APX\setminus (\OPT\cup\calP_{k})}$.} We illustrate this case in Figure~\ref{fig:startingCase1}.
    \begin{figure}[ht]
        \centering
        \includegraphics{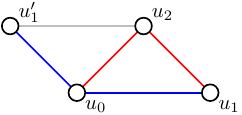}
        \caption{Proof of Lemma~\ref{lem:firstChunk}, Case \textbf{(1)}. We illustrate the case when $u_1u_2\in\APX\setminus\OPT$, the other case (when $u_1u_2\in\APX\cap\OPT$) being similar.}
        \label{fig:startingCase1}
    \end{figure}
    Since $u_0$ is deficient w.r.t. $\calP_{k}$ and $u_0u_2\in\APX\setminus (\OPT\cup\calP_{k})$, there must exist an edge $u_0u_1'\in \OPT\setminus (\APX\cup\calP_{k}), u_1'\neq u_1$. We claim that $P_{k+1}:=u_0u_1'$ satisfies the claim. Indeed, $P_1,\ldots,P_{k+1}$ satisfies Invariants~\ref{inv:main}.\ref{prp:chunks},~\ref{inv:main}.\ref{prp:disjoint} and~\ref{inv:main}.\ref{prp:startingEnding}. Since $u_0u_1'\in \OPT\setminus \APX$, Invariant~\ref{inv:main}.\ref{prp:blueTriangles} ($\OPT$-triangle) is also satisfied. 

    Assume to get a contradiction that Invariant~\ref{inv:main}.\ref{prp:redTriangles} ($\APX$-triangle) is not satisfied. Then, there must be a triangle in $\APX\cup \{u_0u_1'\}$ containing $u_0u_1'$. Since $u_0$ is deficient w.r.t. $\calP_k$, $u_0$ is deficient by Claim~\ref{claim:deficient}. Thus $u_0u_2$ is the only edge of $\APX$ incident to $u_0$. Thus, the only triangle in $\APX\cup \{u_0u_1'\}$ is $u_0u_1'u_2$. But then $N_{\APX}(u_2) = \{u_0, u_1, u_1'\}$, which is a contradiction to the fact that $\APX$ is a $2$-matching. 
    
    \vspace{2.5mm}\noindent\textbf{(2) $\mathbf{u_0u_2\in \APX\cap (\OPT\cup\calP_{k})}$.} We illustrate this case in Figure~\ref{fig:startingCase2}

    \begin{figure}[ht]
        \centering
        \includegraphics{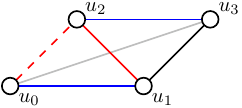}
        \caption{Proof of Lemma~\ref{lem:firstChunk}, Case \textbf{(2)}. We illustrate the case when $u_0u_2\in \APX\cap\calP_{\af{k}}$, the case when $u_0u_2\in \APX\cap\OPT$ being similar.}
        \label{fig:startingCase2}
    \end{figure}
    \begin{Claim}\label{claim:startingCase2u1u2}
        $u_1u_2\in\APX\setminus(\OPT\cup\calP_k)$
    \end{Claim}
    \begin{proof}
        Recall that $u_1u_2\in \APX$. Assume to get a contradiction that $u_1u_2\in \APX\cap\OPT$. If $u_0u_2\in \APX\cap \OPT$ then $u_0u_1u_2$ is a triangle in $\OPT$, a contradiction. Then, since by assumption $u_0u_2\in \APX\cap (\OPT\cup\calP_{k})$, one has $u_0u_2\in \APX\cap\calP_{k}$. Since $u_1u_2\in\APX\cap\OPT$, one has $u_1u_2\in\OPT\setminus\calP_k$. Recall that $u_0u_1\in\OPT\setminus\calP_{k}$. Therefore, $u_0u_1u_2$ is a triangle in $\OPT\triangle\calP_{k}$, a violation of Invariant~\ref{inv:main}.\ref{prp:blueTriangles} ($\OPT$-triangle). Thus $u_1u_2\in \APX\setminus \OPT$.

        Moreover, since either $u_0u_2\in\APX\cap\OPT$ or $u_0u_2\in\APX\cap\calP_k$, and $\calP_k\subseteq\APX\triangle\OPT$, we have that $u_0u_2\in\OPT\triangle\calP_k$. Thus, one has $u_1u_2\notin \calP_{k}$, otherwise $u_0u_1u_2$ would be a triangle in $\OPT\triangle\calP_{k}$, a violation of Invariant~\ref{inv:main}.\ref{prp:blueTriangles} ($\OPT$-triangle).
    \end{proof}

    Consider $c := u_0u_1, u_1u_2$. Notice that $c$ is a chunk. Indeed, $u_0u_1u_2$ is a triangle in $\APX\cup\OPT$, $u_0u_1\in\OPT\setminus\APX$ and by Claim~\ref{claim:startingCase2u1u2}, $u_1u_2\in\APX\setminus\OPT$. 
    
    Suppose that $\OPT\triangle(\calP_{k}\cup\{u_0u_1, u_1u_2\})$ contains no triangle. In this case, $P_{k+1}:=c$ satisfies the claim of the lemma. \fab{Indeed,} every edge of $c$ is free (i.e., not contained in $\calP_k$), so $P_1,\dots, P_{k+1}$ satisfies Invariants~\ref{inv:main}.\ref{prp:chunks},~\ref{inv:main}.\ref{prp:disjoint} and~\ref{inv:main}.\ref{prp:startingEnding}. Invariant~\ref{inv:main}.\ref{prp:blueTriangles} ($\OPT$-triangle) holds by assumption. Recall that $u_0u_1u_2$ is the only triangle in $\APX\cup\{u_0u_1\}$. Then $\APX\triangle(\{u_0u_1, u_1u_2\})=\APX\cup\{u_0u_1\}\setminus\{u_1u_2\}$ contains no triangle, and  Invariant~\ref{inv:main}.\ref{prp:redTriangles} ($\APX$-triangle) is also satisfied.
            
    From now on we assume that there is a triangle in $\OPT\triangle(\calP_{k}\cup\{u_0u_1, u_1u_2\})$. By Invariant~\ref{inv:main}, $\OPT\triangle\calP_{k}$ contains no triangle, and $u_0u_1\in\OPT$, so it must be that such triangle contains the edge $u_1u_2$. By the Unique Triangle Corollary~\ref{cor:uniqueTriangle}, there is only one such triangle. Let $u_1u_2u_3$ be the only triangle in $\OPT\triangle(\calP_{k}\cup\{u_0u_1, u_1u_2\})$. Since $u_0u_1\notin\OPT\triangle(\calP_{k}\cup\{u_0u_1, u_1u_2\})$, then $u_3\neq u_0$.

    \begin{Claim}\label{claim:startingCase2u2u3}
        $u_2u_3\in\OPT\setminus(\APX\cup \calP_k)$
    \end{Claim}
    \begin{proof}
        Since $N_\APX(u_2)=\{u_0, u_1\}$, $u_2u_3\in\OPT\setminus\APX$. Since $u_2u_3\in \OPT\triangle(\calP_{k}\cup\{u_0u_1, u_1u_2\})$, $u_2u_3\notin\calP_k$.
    \end{proof}
    Observe that $c:=u_0u_1,u_1u_2,u_2u_3$ is a chunk: indeed, $u_0u_1u_2$ and $u_1u_2u_3$ are triangles in $\APX\cup\OPT$, $u_0u_1\in \OPT\setminus \APX$, 
    $u_1u_2\in\APX\setminus\OPT$ (by Claim~\ref{claim:startingCase2u1u2}), and $u_2u_3\in\OPT\setminus\APX$ (by Claim ~\ref{claim:startingCase2u2u3}).    

    We \fab{next show} that $P_{k+1}:=c$ satisfies the claim \fab{of the lemma}. Every edge of $c$ is free, so $P_1,\dots, P_{k+1}$ satisfies Invariants~\ref{inv:main}.\ref{prp:chunks},~\ref{inv:main}.\ref{prp:disjoint} and~\ref{inv:main}.\ref{prp:startingEnding}. Recall that $u_1u_2u_3$ is the only triangle in $\OPT\triangle(\calP_{k}\cup\{u_0u_1, u_1u_2\})$. Since $u_2u_3\in\OPT\setminus\APX$, $\OPT\triangle(\calP_{k}\cup c)$ contains no triangle. Thus  Invariant~\ref{inv:main}.\ref{prp:blueTriangles} ($\OPT$-triangle) is also satisfied.
    
    $u_0u_1u_2$ is the only triangle present in $\APX\cup\{u_0u_1\}$, and that triangle is not present in $\APX\triangle c$. Thus, Invariant~\ref{inv:main}.\ref{prp:redTriangles} ($\APX$-triangle) can be violated only if there is a triangle in $\APX\triangle c$ containing $u_2u_3$. Since $\mi{u_0u_2\in\APX}$ and $u_2u_3\in\OPT\setminus\APX$, by the Parity Lemma~\ref{lem:parityLemma} one has $N_{\APX\triangle c}(u_2) = \{u_0, u_3\}$, so the only triangle in $\APX\triangle c$ must be $u_0u_2u_3$. But since $u_0u_1\in\OPT\setminus\APX$, again by the Parity Lemma~\ref{lem:parityLemma}, one has $N_{\APX\triangle c}(u_0) = \{u_1, u_2\}$, so $u_0u_2\notin\APX\triangle c$. Therefore, there is no triangle in $\APX\triangle c$ and  Invariant~\ref{inv:main}.\ref{prp:redTriangles} ($\APX$-triangle) is also satisfied.
\end{proof}
The proof of Lemma \ref{lem:nextChunk} is similar in spirit to the proof of Lemma \ref{lem:firstChunk}, however more involved. We give it in the appendix.

\newpage

\bibliographystyle{plain}
\bibliography{refs}


\appendix

\section{Proof of Lemma \ref{lem:nextChunk}} 
\label{sec:nextChunk}

In this section we prove Lemma \ref{lem:nextChunk}. In order to prove Lemma \ref{lem:firstChunk} it was sufficient to assume Invariant \ref{inv:main}. For the proof of Lemma \ref{lem:nextChunk} we need to choose more carefully the chunk that we add each time. In more detail, we will use the obvious tie breaking rule that enforces the following invariant. 

\begin{invariant}\label{inv:tieBreak}
    Let $P_1,\ldots,P_k$, $1\leq k\leq |\OPT|-|\APX|$, be the collection of alternating trails constructed by our procedure before any update step. Let $P_k = c_1\circ\dots\circ c_l$, where $c_i$ is a chunk for all $i\in[1, l]$. Then for every $i\in[1, l - 1]$ one has:
    \begin{enumerate}
        \item If there exists a chunk $c$ such that $c_1\circ\dots\circ c_i\circ c$ satisfies Invariant~\ref{inv:main} and the first edge of $c$ shares a triangle in $\APX\cup\OPT$ with the last edge of $c_i$, then $c_{i+1}$ shares a triangle in $\APX\cup\OPT$ with the last edge of $c_i$.\label{inv:tieBreak:destroyTriangle}
        \item Suppose there is no chunk satisfying the condition of Invariant~\ref{inv:tieBreak}.\ref{inv:tieBreak:destroyTriangle} and the last edge of $c_i$ is in $\APX$. For every chunk $c$ such that $c_1\circ\dots\circ c_i\circ c$ satisfies Invariant~\ref{inv:main}, $c_{i+1}$ is contained in at most the same number of triangles in $\APX\cup\OPT$ with $2$ edges in $\OPT$ as $c$. \label{inv:tieBreak:avoidTriangleBlue} 
        \item Suppose there is no chunk satisfying the condition of Invariant~\ref{inv:tieBreak}.\ref{inv:tieBreak:destroyTriangle} and the last edge of $c_i$ is in $\OPT$. For every chunk $c$ such that $c_1\circ\dots\circ c_i\circ c$ satisfies Invariant~\ref{inv:main}, $c_{i+1}$ is contained in at most the same number of triangles in $\APX\cup\OPT$ with $2$ edges in $\APX$ as $c$. \label{inv:tieBreak:avoidTriangleRed} 
    \end{enumerate}
\end{invariant}
We will next assume that we are given a collection of trails $P_1,\ldots,P_k$ that satisfy Invariants \ref{inv:main} and \ref{inv:tieBreak}, and we wish to show that we can increase $|\calP_k|$ while maintaining both invariants. Notice that, whenever we start a new trail $P_{k+1}$ by means of Lemma \ref{lem:firstChunk}, then Invariant \ref{inv:tieBreak} is obviously satisfied (since $l=1$ in that case). Therefore it is sufficient to show that we can prove Lemma \ref{lem:nextChunk} under the assumption that initially both invariants hold for the current set of alternating trails. We also observe the following
\begin{remark}\label{rem:notAugmenting}
By construction, given any trail $P_j=c_1\circ\ldots\circ c_l$ in the current set of alternating trails and any $i\in [1,l-1]$, $P':= c_1\circ \ldots\circ c_i$ is not augmenting or it is augmenting but its last node is not deficient. Furthermore, $P_1,\ldots,P_{j-1},P'$ satisf\fab{y} Invariant \ref{inv:main}.
\end{remark}

We will need the following second corollary of the Parity Lemma \ref{lem:parityLemma}.
\begin{corollary}\label{cor:parityCorollary}
    Let $P_1,\ldots,P_{k}$ be a collection of edge-disjoint augmenting trails such that the endpoints of $P_j$ are deficient w.r.t. $\calP_{j-1}$. Then for every  $u\in V$, $j\in[1, k]$, one has (1) $|N_{\APX\triangle P_j}(u)|\leq 2$ and (2) $|N_{\OPT\triangle\calP_j}(u)|\leq 2$.
\end{corollary}
\begin{proof}
    Assume by contradiction that (1) does not hold, i.e. $|N_{\APX\triangle P_j}(u)| > 2$ \mi{for some $u\in V$}. Then $\APX\triangle P_j$ is not a $2$-matching, \af{thus} $P_j$ is not augmenting. \fab{Assume by \af{contradiction that (2)} does not hold, i.e. $|N_{\OPT\triangle \af{\calP}_j}(u)| > 2$} \mi{for some $u\in V$}. $P_j$ satisfies the conditions of Lemma \ref{lem:parityLemma} w.r.t. $P_1,\ldots,P_{j-1}$, \mi{so $u$ is the last node of $P_j$}. Since $u$ is deficient w.r.t. $\calP_{j-1}$, by Claim~\ref{claim:deficient}, it is deficient, so $|N(u)|\leq 3$. If $|N(u)|\leq 2$ the claim is trivial, so assume $|N(u)| = 3$.  Then every edge incident to $u$ belongs to $\OPT\triangle\calP_j$. However the last edge of $P_j$ is in $\OPT$, so it is not in $\OPT\triangle\calP_j$, a contradiction.
\end{proof}
The next lemma shows that there exists a candidate edge to extend the last trail $P_k$.  
\begin{lemma}\label{lem:continueTrail}
 Let $P_1,\ldots,P_{k}$ be a collection of alternating trails satisfying the conditions of Lemma \ref{lem:nextChunk}. Let $u_0u_1$ be the last edge of $P_k$. Then there exist an edge $u_1u_2$ such that $u_1u_2\in(\APX\triangle \OPT)\setminus \calP_k$ and $u_0u_1u_2$ is an alternating trail.
\end{lemma}

\begin{proof}
     We distinguish two cases:

     \vspace{2.5mm}\noindent\textbf{(1) $\mathbf{u_1}$ is not deficient.} Since $u_1$ is not deficient, we have $|N_{\OPT\setminus\APX}(u_1)| = |N_{\APX\setminus\OPT}(u_1)|$. Notice that $|N_{\OPT\setminus\APX}(u_1)| = |N_{\OPT\setminus(\APX\cup \calP_{k})}(u_1)| + |N_{\OPT\cap\calP_k}(u_1)|$ and $|N_{\APX\setminus\OPT}(u_1)| = |N_{\APX\setminus(\OPT\cup \calP_{k})}(u_1)| + |N_{\APX\cap\calP_k}(u_1)|$. Therefore, 
     \begin{align*}
     & |N_{\OPT\setminus(\APX\cup \calP_{k})}(u_1)| - |N_{\APX\setminus(\OPT\cup\calP_k)}(u_1)|\\
     = & |N_{\OPT\setminus\APX}(u_1)| - |N_{\APX\setminus\OPT}(u_1)| - |N_{\OPT\cap\calP_k}(u_1)| + |N_{\APX\cap\calP_k}(u_1)|\\
     = & |N_{\APX\cap\calP_k}(u_1)| - |N_{\OPT\cap\calP_k}(u_1)|.
     \end{align*}
    
     By Claim~\ref{claim:deficient}, $u_1$ is not deficient w.r.t. $\calP_j$ for any $j\in[1, k]$, so $u_1$ is not an endpoint of any trail $P_1,\dots, P_{k-1}$ nor the starting node of $P_k$. Any pair of edges of type $vu_1,u_1w$ along any trail $P_j$, $j\in[1, k]$, contributes with $0$ to $|N_{\APX\cap\calP_k}(u_1)| - |N_{\OPT\cap\calP_k}(u_1)|$. The last edge of $P_k$ contributes with $+1$ to $|N_{\APX\cap\calP_k}(u_1)| - |N_{\OPT\cap\calP_k}(u_1)|$ if $u_0u_1\in\APX\setminus\OPT$ and with $-1$ if $u_0u_1\in\OPT\setminus\APX$. The remaining nodes of $N(u_1)$ contribute $0$ to $|N_{\APX\cap\calP_k}(u_1)| - |N_{\OPT\cap\calP_k}(u_1)|$. Therefore, if $u_0u_1\in\APX\setminus\OPT$, then $|N_{\OPT\setminus(\APX\cup \calP_{k})}(u_1)| - |N_{\APX\setminus(\OPT\cup\calP_k)}(u_1)| = 1$ and one has $|N_{\OPT\setminus(\APX\cup \calP_{k})}(u)|\geq 1$. If $u_0u_1\in\OPT\setminus\APX$, then $|N_{\OPT\setminus(\APX\cup \calP_{k})}(u_1)| - |N_{\APX\setminus(\OPT\cup\calP_k)}(u_1)| = -1$ and one has\newline $|N_{\APX\setminus(\OPT\cup\calP_k)}(u)| \geq 1$. The claim of the lemma follows.

    \vspace{2.5mm}\noindent\textbf{(2) $\mathbf{u_1}$ is deficient.} If $|N_{\APX\setminus\OPT}(u_1)| = 0$, then $u_0u_1\in\OPT\cap P_k$, $u_1$ is deficient w.r.t. $\calP_{k-1}$ and $P_k$ is augmenting, a contradiction to the assumption of the lemma. Therefore, $|N_{\APX\setminus\OPT}(u_1)| \geq 1$. Since $u_1$ is deficient one has $|N_{\APX\setminus\OPT}(u_1)| = 1$ and $|N_{\OPT\setminus\APX}(u_1)| = 2$.
    \begin{Claim}\label{claim:continueTrailConsecutive}
    \mi{For every $j\in[1, k]$, there are \af{no}  $2$ consecutive edges of $P_j$ incident to $u_1$.}   
    \end{Claim}
    \begin{proof}
    \mi{By contradiction, assume there are two edges $vu_1, u_1w$ in $P_j$ for some $j\in[1, k]$. Assume w.l.o.g. $vu_1\in\APX\setminus\OPT, u_1w\in\OPT\setminus\APX$. Since $u_0u_1$ is the last edge of $P_k$, it must be that $u_0\notin\{v, w\}$. Thus, using that $|N_{\OPT\setminus\APX}(u)| = 2, |N_{\APX\setminus\OPT}(u)| = 1$, one has $u_0u_1\in\OPT\cap P_k$, $u_1$ is deficient w.r.t. $\calP_{k-1}$, and $P_k$ is augmenting, a contradiction to the assumptions of the lemma.}\mig{Changed this from intermediate now to a more precise claim. The proof is almost the same as before, but more clear now.} 
    \end{proof}
    

     \begin{Claim}
        Excluding as the last node of $P_k$, $u_1$ appears at most once as the endpoint of any $P_1,\dots, P_k$.
     \end{Claim}
     \begin{proof}
     \mi{The first and last edges of every trail $P_1,\dots, P_{k-1}$ are in $\OPT\setminus\APX$. Thus, by Claim~\ref{claim:continueTrailConsecutive}, the edge of $\APX\setminus\OPT$ incident to $u_1$ is not used by any trail $P_1,\dots, P_{k-1}$.}\mig{Changed this to match the updated claim above.} Therefore, $u_1$ is not simultaneously the first and last node of any $P_j$, $j\in[1, k-1]$, because otherwise $|N_{\APX\triangle P_j}(u_1)| = 3$, implying that $\APX\triangle P_j$ is not a $2$-matching: this contradicts the fact that $P_j$ is an augmenting trail.

     Let $j$ be the smallest index such that $u_1$ is an endpoint of $P_j$. If $j=k$ the claim trivially holds. Otherwise, $u_1$ is not deficient w.r.t $\calP_j$: indeed, $|\OPT\triangle \calP_j|\leq 1$ and $|\APX\triangle \calP_j|=1$. Thus $u_1$ cannot be the endpoint of any $P_{j+1},\ldots,P_{k-1}$, nor the first node of $P_k$.
    \end{proof}
    By the above two claims, excluding $u_0u_1$, at most one edge incident to $u_1$ is used by the trails $P_1,\dots, P_k$, and that edge, if any, is in $\OPT\setminus\APX$. We conclude that if $u_0u_1\in\APX\setminus\OPT$, then $|N_{\OPT\setminus(\APX\cup\calP_k)}(u_1)|\geq 1$, and if $u_0u_1\in\OPT\setminus\APX$, then $|N_{\APX\setminus(\OPT\cup\calP_k)}(u_1)| = 1$. The claim of the lemma follows.
\end{proof}

The proof of Lemma \ref{lem:nextChunk} is split into two parts which use very similar arguments: in Lemma \ref{lem:extendingTrailBlue} we assume that $P_k$ ends with an edge of $\APX$, while in Lemma \ref{lem:extendingTrailRed} that $P_k$ ends with an edge in $\OPT$.

\begin{lemma}\label{lem:extendingTrailBlue}
    Let $P_1,\dots, P_k$ be the current collection of alternating trails (satisfying Invariants \ref{inv:main} and \ref{inv:tieBreak}), where $P_k$ is not augmenting or the last node of $P_k$ is not deficient. Suppose that the last edge $u_0u_1$ of $P_k$ is in $\APX$. Then there exist a chunk \fab{$c$} such that $P_1,\ldots,P_{k-1},P_k\circ c$ satisfy Invariants~\ref{inv:main} and \ref{inv:tieBreak}.
\end{lemma}

\begin{proof}      
It is sufficient to show that there exists one such chunk $c$ so that Invariant \ref{inv:main} holds. Invariant \ref{inv:tieBreak} can then be trivially enforced by breaking ties properly.

By Lemma~\ref{lem:continueTrail}, there exists an edge $u_1u_2\in \OPT\setminus(\APX\cup\calP_k)$. Notice that $u_1u_2$ is a chunk. If $c := u_1u_2$ satisfies the claim we are done, hence assume this is not the case. Observe that $P_1,\dots, P_{k-1},P_k\circ c$ satisfies Invariant~\ref{inv:main}.\ref{prp:chunks},~\ref{inv:main}.\ref{prp:disjoint} and~\ref{inv:main}.\ref{prp:startingEnding}. Since $u_1u_2\in \OPT\setminus \APX$ and $\OPT\triangle\calP_k$ is triangle-free, Invariant~\ref{inv:main}.\ref{prp:blueTriangles} ($\OPT$-triangle) is also satisfied. Thus Invariant~\ref{inv:main}.\ref{prp:redTriangles} ($\APX$-triangle) must be violated. In more detail, $\APX\triangle (P_k\cup\{u_1u_2\}) = (\APX\triangle P_k)\cup \{u_1u_2\}$ must contain a triangle (and $u_1u_2$ must belong to that triangle since $\APX\triangle P_k$ is triangle-free). By the Unique Triangle Corollary~\ref{cor:uniqueTriangle}, there is only one such triangle: let $u_1u_2u_3$ be the only triangle in $\APX\triangle (P_k\cup\{u_1u_2\})$. Since $u_0u_1\in \APX\cap P_k$, $u_0u_1\notin \APX\triangle (P_k\cup\{u_1u_2\})$. Thus, $u_0, u_1, u_2, u_3$ are all distinct.

    We consider the following cases depending on the edge $u_1u_3$. Notice that, since $u_1u_3\in \APX\triangle (P_k\cup\{u_1u_2\})$ then $u_1u_3$ is in either $\APX\setminus P_k$ or in $P_k\setminus \APX=\OPT\cap P_k$. In the first case we further distinguish between $u_1u_3\in \APX\setminus(\OPT\cup\calP_k)$ and $u_1u_3\in \APX\cap(\OPT\cup\calP_{k-1})$, leading to $3$ cases in total.
        
    \vspace{2.5mm}\noindent\textbf{(1) $\mathbf{u_1u_3\in \APX\setminus (\OPT\cup\calP_k)}$.} This case is illustrated in Figure~\ref{fig:redCase1}.
    \begin{figure}[ht]
        \centering
        \includegraphics{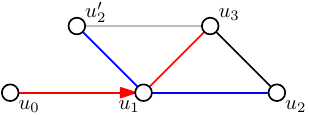}
        \caption{Proof of Lemma~\ref{lem:extendingTrailBlue}, Case \textbf{(1)}.}
        \label{fig:redCase1}
    \end{figure}
    \begin{Claim}\label{claim:blueCase1u1u2'}
        There exists an edge $u_1u_2'\in\OPT\setminus(\APX\cup\calP_k), u_2'\neq u_2$.
    \end{Claim}
    \begin{proof}
        By the Degree Domination Lemma~\ref{lem:blueMajority}, $|N_\OPT(u_1)|\geq |N_\APX(u_1)|$. Since $N_{\APX\setminus \OPT}(u_1) = \{u_0, u_3\}$, there must exist an edge $u_1u_2'\in \OPT\setminus\APX, u_2'\neq u_2$. 
        
        Recall that $u_0u_1\in \APX\cap P_k$, so, by Invariant~\ref{inv:main}.\ref{prp:disjoint}, $u_0u_1\in\APX\setminus\calP_{k-1}$.\fabr{Rephrased a bit, check} Assume by contradiction that $u_1u_2'\in P_j$ for some $j\in [1,k-1]$. Since $u_0u_1\in\APX\setminus \calP_{k-1}$ and $u_1u_3\in\APX\setminus\calP_k$, one has $N_{\APX\triangle P_j}(u_1)\supseteq \{u_0, u_2', u_3\}$, a contradiction to Corollary~\ref{cor:parityCorollary}. Therefore $u_1u_2'\notin\fab{\calP_{k-1}}$.     
        Also, $u_1u_2'\notin P_k$, otherwise, since $u_1u_2\in\OPT\setminus\APX, u_1u_3\in\APX\setminus\calP_k$, $N_{\APX\triangle(P_k\cup\{u_1u_2\})}(u_1)\supseteq \{u_2, u_2', u_3\}$, a contradiction to the Parity Lemma~\ref{lem:parityLemma}. The claim follows.
    \end{proof}
    
    By Claim~\ref{claim:blueCase1u1u2'}, $c:=u_1u_2'$ is a chunk. We claim that $P_1,\dots, P_k\circ c$ satisf\fab{y} Invariant~\ref{inv:main}. Observe that  Invariant\fab{s}~\ref{inv:main}.\ref{prp:chunks},~\ref{inv:main}.\ref{prp:disjoint} and~\ref{inv:main}.\ref{prp:startingEnding} are satisfied. Since $u_1u_2\in \OPT\setminus \APX$ and $\OPT\triangle\calP_k$ is triangle-free, Invariant~\ref{inv:main}.\ref{prp:blueTriangles} ($\OPT$-triangle) is also satisfied. It is only left to prove that Invariant~\ref{inv:main}.\ref{prp:redTriangles} ($\APX$-triangle) holds. Assume to get a contradiction that there is a triangle in $\APX\triangle (P_k\cup\{u_1u_2'\})$.
    
    Since $\APX\triangle P_k$ contains no triangle, every triangle in $\APX\triangle (P_k\cup\{u_1u_2'\})$ must contain the edge $u_1u_2'$. Since $u_1u_3\in\APX\setminus\calP_k$ and, by Claim~\ref{claim:blueCase1u1u2'}, $u_1u_2'\in\OPT\setminus(\APX\cup\calP_k)$, by the Parity Lemma~\ref{lem:parityLemma} \fab{($u_1$ is not the last node of the trail $P_k\circ u_1u_2'$)}\mig{I moved this remainder to the first time the parity lemma is used (and modified it accordingly)} one has $N_{\APX\triangle(P_k\cup \{u_1u_2'\})}(u_1) = \{u_2', u_3\}$. Thus, the only triangle in $\APX\triangle(P_k\cup\{u_1u_2'\})$ is $u_1u_2'u_3$. 
    
    Since $u_1u_2u_3$ is a triangle in $\APX\triangle (P_k\cup\{u_1u_2\})$, by the Parity Lemma~\ref{lem:parityLemma}, one has $N_{\APX\triangle (P_k\cup\{u_1u_2\})}(u_3) = \{u_1, u_2\}$. Observing that $N_{\APX\triangle (P_k\cup\{u_1u_2\})}(u_3) = N_{\APX\triangle (P_k\cup\{u_1u_2'\})}(u_3)$ we see that $u_2'u_3\notin \APX\triangle (P_k\cup\{\af{u_1u_2'}\})$, a contradiction (since $u_1u_2'u_3$ is a triangle in $\APX\triangle(P_k\cup\{u_1u_2'\}$).

    \vspace{2.5mm}\noindent\textbf{(2) $\mathbf{u_1u_3\in \APX\cap (\OPT\cup\calP_{k-1})}$.} This case is illustrated in Figure~\ref{fig:redCase2}.
    \begin{figure}[ht]
        \centering
        \includegraphics{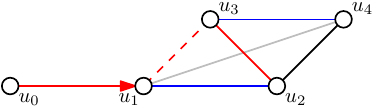}
        \caption{Proof of Lemma~\ref{lem:extendingTrailBlue}, Case \textbf{(2)}. We illustrate the case when $u_1u_3\in \APX\cap\calP_{k-1}$, the other case (when $u_1u_3\in\APX\cap\OPT$) being similar.}
        \label{fig:redCase2}
    \end{figure}
    \begin{Claim}\label{claim:blueCase2u2u3}
        $u_2u_3\in\APX\setminus(\OPT\cup\calP_k)$.
    \end{Claim}
    \begin{proof}
        Recall that $u_1u_2\in\OPT\setminus(\APX\af{\cup}\calP_k)$. Assume to get a contradiction that $u_2u_3\in \OPT$. It must be that $u_1u_3\notin \OPT$, otherwise $u_1u_2u_3$ would be a triangle in $\OPT$, so by Case~\textbf{(2)} assumption $u_1u_3\in\APX\cap\calP_{k-1}$. Also, $u_2u_3\in \OPT\cap P_k$, because $u_1u_2u_3$ is a triangle in $\APX\triangle(P_k\cup\{u_1u_2\})$. One has $u_1u_2\notin\calP_k$, and since $u_2u_3\in P_k$, $u_2u_3\notin\calP_{k-1}$. Then, $u_1u_2u_3$ is a triangle in $\OPT\triangle\calP_{k-1}$, a violation of Invariant~\ref{inv:main}.\ref{prp:blueTriangles} ($\OPT$-triangle) w.r.t. $\calP_{k-1}$. 

         \afr{It took me a while to  understand this paragraph. I Expanded the argument. Please Check.} So far we proved that $u_2u_3\in\APX\setminus\OPT$. 
        \af{Clearly $u_1u_2\in \OPT\triangle\calP_k$. By assumption of Case~\textbf{(2)}, $u_1u_3\in \APX\cap (\OPT\cup \calP_{k-1})$, so either $u_1u_3\in\APX\cap\OPT$ or $u_1u_3\in\APX\cap\calP_{k-1}$. Since $\calP_{k-1}\subset\calP_k\subseteq \APX\triangle \OPT$, $u_1u_3\in \OPT\triangle\calP_k$. Note that if $u_2u_3\in\APX\cap\calP_k$, then }$u_1u_2u_3$ is a triangle in $\OPT\triangle\calP_k$, a violation of Invariant~\ref{inv:main}.\ref{prp:blueTriangles} ($\OPT$-triangle), so $u_2u_3\notin\calP_k$. The claim follows.
    \end{proof}

    Recall that $u_1u_2\in\OPT\setminus(\APX\cup\calP_k)$. By Claim~\ref{claim:blueCase2u2u3}, $u_2u_3\in\APX\setminus(\OPT\cup\calP_k)$. Also, $u_1u_2u_3$ is a triangle in $\APX\cup\OPT$, so $c:=u_1u_2, u_2u_3$ is a chunk.

    Suppose that $\OPT\triangle(\calP_k\cup\{u_1u_2, u_2u_3\})$ contains no triangle. In this case $c$ satisfies the claim of the Lemma. Indeed $P_1,\dots, P_k\circ c$ satisfy Invariants~\ref{inv:main}.\ref{prp:chunks},~\ref{inv:main}.\ref{prp:disjoint} and~\ref{inv:main}.\ref{prp:startingEnding}. Invariant~\ref{inv:main}.\ref{prp:blueTriangles} ($\OPT$-triangle) holds by assumption. Recall that $u_1u_2u_3$ is the only triangle in $\APX\triangle (P_k\cup\{u_1u_2\})$. Since $u_2u_3\in\APX$, $\APX\triangle (P_k\cup\{u_1u_2, u_2u_3\})$ contains no triangle and Invariant~\ref{inv:main}.\ref{prp:redTriangles} ($\APX$-triangle) is also satisfied.


    From now on we assume that there is a triangle in $\OPT\triangle(\calP_k\cup\{u_1u_2, u_2u_3\})$. By Invariant~\ref{inv:main}.\ref{prp:blueTriangles} ($\OPT$-triangle), $\OPT\triangle\calP_k$ contains no triangles. Since $u_1u_2\in\OPT$, every triangle in $\OPT\triangle(\calP_k\cup\{u_1u_2, u_2u_3\})$ must contain the edge $u_2u_3$. By the Unique Triangle Corollary~\ref{cor:uniqueTriangle}, there is only one such triangle: let $u_2u_3u_4$ be the only triangle in $\OPT\triangle(\calP_k\cup\{u_1u_2, u_2u_3\})$. Since $u_1u_2\in \OPT$, $u_1u_2\notin \OPT\triangle(\calP_k\cup\{u_1u_2, u_2u_3\})$, so $u_4\neq u_1$ and thus $u_1, u_2, u_3, u_4$ are all distinct. 

    \begin{Claim}\label{claim:blueCase2u3u4}
        $u_3u_4\in \OPT\setminus(\APX\cup\calP_k).$
    \end{Claim}
    \begin{proof}
         As $N_{\APX}(u_3) = \{u_1, u_2\}$, we have $u_3u_4\in \OPT\setminus\APX$. Furthermore, since $u_3u_4\in\OPT\triangle(\calP_k\cup  \{u_1u_2, u_2u_3\})$, then $u_3u_4\notin\calP_k$. 
    \end{proof}
    
    Recall that $u_1u_2\in\OPT\setminus(\APX\cup\calP_k)$. By Claims~\ref{claim:blueCase2u2u3} and~\ref{claim:blueCase2u3u4}, $u_2u_3\in\APX\setminus(\OPT\cup\calP_k)$ and $u_3u_4\in\OPT\setminus(\APX\cup\calP_k)$. Furthermore, $u_1u_2u_3$ and $u_2u_3u_4$ are triangles in $\APX\cup\OPT$, so $c := u_1u_2, u_2u_3, u_3u_4$ is a chunk. We next show that $c$ satisfies the claim of the Lemma. 
    
    $P_1,\dots, P_k\circ c$ satisfy Invariants~\ref{inv:main}.\ref{prp:chunks},~\ref{inv:main}.\ref{prp:disjoint} and~\ref{inv:main}.\ref{prp:startingEnding}. Recall that $u_2u_3u_4$ is the only triangle in $\OPT\triangle(\calP_k\cup\{u_1u_2, u_2u_3\})$. Since $u_3u_4\in\OPT$, $\OPT\triangle(\calP_k\cup c)$ contains no triangle, so Invariant~\ref{inv:main}.\ref{prp:blueTriangles} ($\OPT$-triangle) is also satisfied. It remains to prove that  Invariant~\ref{inv:main}.\ref{prp:redTriangles} ($\APX$-triangle) holds. Assume to get a contradiction that there is a triangle in $\APX\triangle (P_k\cup c)$. 
    
    Recall that $u_1u_2u_3$ is the only triangle in $\APX\triangle (P_k\cup\{u_1u_2\})$, and, since $u_2u_3\in\APX$, it is not present in $\APX\triangle (P_k\cup c)$. Thus, every triangle in $\APX\triangle (P_k\cup c)$ must contain the edge $u_3u_4$. 

    We have $u_3u_4\in \OPT\setminus\APX$ and \af{ $u_1u_3\in \APX\triangle P_k$,} so by the Parity Lemma~\ref{lem:parityLemma}, $N_{\APX\triangle (P_k\cup c)}(u_3) = \{u_1, u_4\}$. Thus, the only possible triangle in $\APX\triangle (P_k\cup c)$ is $u_1u_3u_4$. 
    Since \af{$u_1u_2u_3$ is a triangle in $\APX\triangle (P_k\cup \{u_1u_2\})$}, one has $N_{\APX\triangle (P_k\cup c)}(u_1) = \{u_2, u_3\}$, implying $u_1u_4\notin \APX\triangle (P_k\cup c)$, a contradiction. 
    
    \vspace{2.5mm}\noindent\textbf{(3) $\mathbf{u_1u_3\in \OPT\cap P_k}$.}

    \begin{Claim}\label{claim:blueCase3u2u3}
        $u_2u_3\in \APX\setminus(\OPT\cup \calP_k)$. 
    \end{Claim}
    \begin{proof}
        Recall that $u_1u_2\in\OPT\setminus(\APX\cup\calP_k)$. Since $u_1u_2, u_1u_3\in \OPT\setminus \APX$, it must be that $u_2u_3\in \APX\setminus \OPT$, otherwise $u_1u_2u_3$ would be a triangle in $\OPT$. Since $u_1u_2u_3$ is a triangle in $\APX\triangle(P_k\cup\{u_1u_2\})$, then $u_2u_3\notin P_k$. By Invariant~\ref{inv:main}.\ref{prp:disjoint} and since $u_1u_3\in P_k$, $u_1u_3\notin\calP_{k-1}$. One has $u_1u_2, u_1u_3\in\OPT\setminus\calP_{k-1}$. If $u_2u_3\in\calP_{k-1}$, then $u_1u_2u_3$ is a triangle in $\OPT\triangle\calP_{k-1}$ a violation of Invariant~\ref{inv:main}.\ref{prp:blueTriangles} ($\OPT$-triangle) w.r.t. $\calP_{k-1}$. The claim follows.
    \end{proof}
    
    \af{As $u_1u_2u_3$ is a triangle in $\APX\cup\OPT$ , $u_1u_2\in\OPT\setminus(\APX\cup\calP_k)$ and since by Claim~\ref{claim:blueCase3u2u3} $u_2u_3\in\APX\setminus(\OPT\cup\calP_k)$ then $\fab{c:=}u_1u_2, u_2u_3$ is a chunk.}

    Suppose that $\OPT\triangle(\calP_k\cup\{u_1u_2, u_2u_3\})$ contains no triangle. Then $c$ satisfies the claim of the lemma by the same argument used in case (2) (with Claim~\ref{claim:blueCase3u2u3} in place of Claim~\ref{claim:blueCase2u2u3}).
    
    
    From now on we assume that there is a triangle in $\OPT\triangle(\calP_k\cup\{u_1u_2, u_2u_3\})$. By Invariant~\ref{inv:main}.\ref{prp:blueTriangles} ($\OPT$-triangle), $\OPT\triangle\calP_k$ contains no triangles. Since $u_1u_2\in\OPT$, every triangle in $\OPT\triangle(\calP_k\cup\{u_1u_2, u_2u_3\})$ must contain the edge $u_2u_3$. By the Unique Triangle Corollary~\ref{cor:uniqueTriangle}, there is only one such triangle: let $u_2u_3u_4$ be the only triangle in $\OPT\triangle(\calP_k\cup\{u_1u_2, u_2u_3\})$. Since $u_1u_2\in \OPT$, $u_1u_2\notin \OPT\triangle(\calP_k\cup\{u_1u_2, u_2u_3\})$, so $u_4\neq u_1$ and thus $u_1, u_2, u_3, u_4$ are all distinct. We consider two cases:
    
    \vspace{2.5mm}\noindent\textbf{(3.1) $\mathbf{u_3u_4\in \OPT}$.} We illustrate this case in Figure~\ref{fig:redCase31}.  \begin{figure}[ht]
        \centering
        \includegraphics{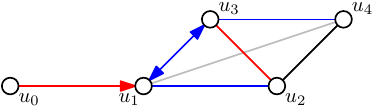}
        \caption{Proof of Lemma~\ref{lem:extendingTrailBlue}, Case \textbf{(3.1)}.}
        \label{fig:redCase31}
    \end{figure}
    \begin{Claim}\label{claim:blueCase3u3u4}
        $u_3u_4\in \OPT\setminus (\APX\cup\calP_k)$. 
    \end{Claim}
    \begin{proof}
        Since $u_3u_4\in \OPT\triangle(\calP_k\cup\{u_1u_2, u_2u_3\})$, it must be that $u_3u_4\notin\calP_k$. Assume by contradiction that $u_3u_4\in \APX\cap \OPT$. By assumption of this case, $u_1u_3\in\OPT\cap P_k$ and by Claim~\ref{claim:blueCase3u2u3}, $u_2u_3\in\APX\setminus(\OPT\cup\calP_k)$. Thus, one has $N_{\APX\triangle P_k}(u_3) \supseteq \{u_1, u_2, u_4\}$, a contradiction by the Parity Lemma~\ref{lem:parityLemma}. The claim follows.
    \end{proof}
    
    Recall that $u_1u_2\in\OPT\setminus(\APX\cup\calP_k)$, and by Claims~\ref{claim:blueCase3u2u3} and~\ref{claim:blueCase3u3u4}, $u_2u_3\in\APX\setminus(\OPT\cup\calP_k)$ and $u_3u_4\in\OPT\setminus(\APX\cup\calP_k)$. Furthermore, $u_1u_2u_3$ and $u_2u_3u_4$ are triangles in $\APX\cup\OPT$, so $c:=u_1u_2,u_2u_3,u_3u_4$ is a chunk. We claim that $c$ satisfies the claim of the lemma. 

    Observe that $P_1,\dots, P_k\circ c$ satisfy Invariants~\ref{inv:main}.\ref{prp:chunks},~\ref{inv:main}.\ref{prp:disjoint} and~\ref{inv:main}.\ref{prp:startingEnding}. Recall that $u_2u_3u_4$ is the only triangle in $\OPT\triangle(\calP_k\cup\{u_1u_2, u_2u_3\})$. Since $u_3u_4\in \OPT$, $\OPT\triangle(\calP_k\cup c)$ contains no triangle, so Invariant~\ref{inv:main}.\ref{prp:blueTriangles} ($\OPT$-triangle) is also satisfied. It remains to show that  Invariant~\ref{inv:main}.\ref{prp:redTriangles} ($\APX$-triangle) holds. Assume to get a contradiction that there is a triangle in $\APX\triangle(P_k\cup c)$.

    Recall that $u_1u_2u_3$ is the only triangle in $\APX\triangle (P_k\cup\{u_1u_2\})$. Since $u_2u_3\in \APX$, $u_2u_3\notin \APX\triangle (P_k\cup c)$. Thus, every triangle in $\APX\triangle (P_k\cup c)$ must contain the edge $u_3u_4$. 

    We have $u_3u_4\in \OPT\setminus\APX$ and $u_1u_3\in \OPT\cap P_k$, so by the Parity Lemma~\ref{lem:parityLemma}, $N_{\APX\triangle (P_k\cup c)}(u_3) = \{u_1, u_4\}$. Thus, the only possible triangle in $\APX\triangle (P_k\cup c)$ is $u_1u_3u_4$. Since $u_1u_2\in\OPT\setminus\APX$, by the Parity Lemma~\ref{lem:parityLemma}, one has $N_{\APX\triangle (P_k\cup c)}(u_1) = \{u_2, u_3\}$, implying $u_1u_4\notin \APX\triangle (P_k\cup c)$, a contradiction. 

    \vspace{2.5mm}\noindent\textbf{(3.2) $\mathbf{u_3u_4\in \APX\setminus \OPT}$.} We will show that this case is not possible. See Figure~\ref{fig:redCase32} for an illustration.
    \begin{figure}
        \centering
        \includegraphics{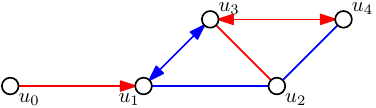}
        \caption{Proof of Lemma~\ref{lem:extendingTrailBlue}, Case \textbf{(3.2)}.}
        \label{fig:redCase32}
    \end{figure}
    
    \begin{Claim}\label{claim:blueCase3u2u4}
        $u_2u_4\in \OPT\setminus(\APX\cup\calP_k)$.
    \end{Claim}
    \begin{proof}
        By Claim~\ref{claim:blueCase3u2u3} and by assumption of this case, $u_2u_3, u_3u_4\in\APX$, so $u_2u_4\in\OPT\setminus\APX$, otherwise $u_2u_3u_4$ would be a triangle in $\APX$. Recall that $u_2u_3u_4$ is a triangle in $\OPT\triangle(\calP_k\cup\{u_1u_2, u_2u_3\})$. Thus $u_2u_4\notin\calP_k$. The claim follows.
    \end{proof}

    \begin{Claim}\label{claim:blueCase3onlyPk}
        $N_{P_k}(u_3) = \{u_1, u_4\}$.
    \end{Claim}
    \begin{proof}
        Since, $u_1u_2u_3$ is a triangle in $\APX\triangle P_k$, by the Parity Lemma~\ref{lem:parityLemma} one has $N_{\APX\triangle P_k}(u_3)=\{u_1, u_2\}$. Therefore, $u_3u_4\notin\APX\triangle P_k$. By assumption of this case $u_3u_4\in\APX\setminus\OPT$, so it must be that $u_3u_4\in P_k$. Recall that $u_1u_3\in\OPT\cap P_k$. Thus, $N_{P_k}(u_3) \supseteq \{u_1, u_4\}$. By Claim~\ref{claim:blueCase3u2u3}, $u_2u_3\notin P_k$. If $u_3$ has another neighbor $v\notin\{u_1, u_2, u_4\}$ in $\APX\cup \OPT$ it must be that $u_3v\in \OPT\setminus \APX$, because $N_\APX(u) = \{u_2, u_4\}$. Moreover, since $N_{\APX\triangle P_k}(u_3) = \{u_1, u_2\}$, it must be that $u_3v\notin P_k$. The claim follows.
    \end{proof}

    \begin{Claim}
        $u_4\neq u_0$.
    \end{Claim}
    \begin{proof}
        By Claim~\ref{claim:blueCase3u2u4} $u_2u_4\in\OPT\setminus\calP_k$, and by assumption Case~\textbf{(3.2)} and Claim~\ref{claim:blueCase3onlyPk}, $u_3u_4\in\APX\cap P_k$. Assume to get a contradiction $u_0=u_4$. There is an edge $u_0u_3\in \APX\cap P_k$ and an edge $u_0u_2\in \OPT\setminus\calP_k$. But then, since $u_0u_1\in\APX\cap P_k$, $N_{\OPT\triangle \calP_k}(u_0) \supseteq \{u_1, u_2, u_3\}$, a contradiction by the Parity Lemma~\ref{lem:parityLemma}.
    \end{proof}

    \begin{Claim}\label{claim:blueCase3u1u4}
        There is no edge $u_1u_4$ in $\APX\cup \OPT$.
    \end{Claim}
    \begin{proof}
        Assume to get a contradiction that $u_1u_4\in \APX\cup \OPT$. Then since $N_\OPT(u_1)=\{u_2, u_3\}$ it must be that $u_1u_4\in\APX\setminus\OPT$. By Claim~\ref{claim:blueCase3u2u4} $u_2u_4\in\OPT\setminus\calP_k$, and by assumption of this case and Claim~\ref{claim:blueCase3onlyPk}, $u_3u_4\in\APX\cap P_k$. If $u_1u_4\in \calP_k$ then $N_{\OPT\triangle\calP_k}(u_4) \supseteq \{u_1, u_2, u_3\}$, a contradiction by the Parity Lemma~\ref{lem:parityLemma}. 
        
        Recall that $u_1u_2\in\OPT\setminus(\APX\cup\calP_k)$. By assumption of this case, $u_1u_3\in\OPT\cap P_k$. If $u_1u_4\notin\calP_k$ then, $N_{\APX\triangle (P_k\cup \{u_1u_2\})}(u_1) \supseteq \{u_2, u_3, u_4\}$, a contradiction by the Parity Lemma~\ref{lem:parityLemma}. The claim follows.
    \end{proof}
    
    By Invariant~\ref{inv:main}.\ref{prp:chunks}, $P_k = c_1\circ\dots\circ c_l$, where $c_j$ is a chunk for all $j\in[1, l]$. By Claim~\ref{claim:blueCase3u1u4}, there is no edge $u_1u_4$ in $\APX\cup\OPT$, so there is no triangle $u_1u_3u_4$ in $\APX\cup\OPT$. By Claim~\ref{claim:blueCase3onlyPk}, $N_{P_k}(u_3)=\{u_1, u_4\}$. Therefore, by Definition~\ref{def:chunk} of chunk, the edges $u_1u_3$ and $u_3u_4$ must belong to consecutive chunks in $P_k$. Thus, we have $P_k = c_1\circ\dots\circ c_i\circ c_{i+1}\circ\dots c_l$, where $c_i$ is either a chunk whose last edge is $u_1u_3$ (and the first edge of $c_{i+1}$ is $u_3u_4$) or $c_i$ is a chunk whose last edge is $u_3u_4$ (and the first edge of $c_{i+1}$ is $u_1u_3$). We next consider these two cases separately.

    \vspace{2.5mm}\noindent\textbf{(3.2.1) $\mathbf{u_1u_3}$ is the last edge of $\mathbf{c_i}$.} This case is illustrated in Figure~\ref{fig:redCase321}.
    \begin{figure}[ht]
        \centering
        \includegraphics{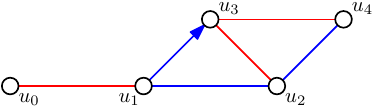}
        \caption{Proof of Lemma~\ref{lem:extendingTrailBlue}, Case \textbf{(3.2.1)}. Here the arrows are used for edges of $P$ instead of edges of $P_k$.}
        \label{fig:redCase321}
    \end{figure}    
    \fab{By Remark~\ref{rem:notAugmenting}, $P_1,\ldots,P_{k-1},P$, where $P:=c_1\circ\dots\circ c_i$}, satisfy Invariant~\ref{inv:main}. By Claim~\ref{claim:blueCase3u2u3}, $u_2u_3\in\APX\setminus(\OPT\cup\calP_k)$, so $c:=u_2u_3$ is a chunk. We claim that\fabr{I corrected this in several places: the invariant is about a set of paths} \fab{$P_1,\ldots,P_{k-1},P'$, where $P':=c_1\circ\dots\circ c_i\circ c$,} satisfy Invariant~\ref{inv:main}. That would lead to a contradiction by Invariant~\ref{inv:tieBreak}.\ref{inv:tieBreak:destroyTriangle}. Indeed, $u_1u_2u_3$ is a triangle in $\APX\cup\OPT$, so the first edge of $c$ (namely, $u_2u_3$) shares a triangle in $\APX\cup\OPT$ with the last edge of $c_i$ (namely, $u_1u_3$). On the other hand, since by Claim~\ref{claim:blueCase3u1u4}, $u_1u_4\notin \fab{\APX\cup \OPT}$, the first edge of $c_{i+1}$ (namely, $u_3u_4$), does not share a triangle in $\APX\cup\OPT$ with the last edge of $c_i$ (namely, $u_1u_3$).

    Observe that \fab{$P_1,\ldots,P_{k-1},P'$} satisfy Invariants~\ref{inv:main}.\ref{prp:chunks},~\ref{inv:main}.\ref{prp:disjoint} and~\ref{inv:main}.\ref{prp:startingEnding}. Since $u_2u_3\in \APX\setminus \OPT$, and $\APX\triangle P$ is triangle-free, Invariant~\ref{inv:main}.\ref{prp:redTriangles} ($\APX$-triangle) is also satisfied. It remains to prove that Invariant~\ref{inv:main}.\ref{prp:blueTriangles} ($\OPT$-triangle) holds. Assume to get a contradiction that $\OPT\triangle(P\cup\calP_{k-1}\cup\{u_2u_3\})$ contains a triangle.

    Since \fab{$P_1,\ldots,P_{k-1},P$} satisfy Invariant~\ref{inv:main}.\ref{prp:blueTriangles} ($\OPT$-triangle), $\OPT\triangle(P\cup\calP_{k-1})$ contains no triangle. Therefore, every triangle in $\OPT\triangle(P\cup\calP_{k-1}\cup\{u_2u_3\})$ must contain the edge $u_2u_3$. Recall that $u_1u_2\in\OPT\setminus \calP_k$. By Claim~\ref{claim:blueCase3u2u4}, $u_2u_4\in\OPT\setminus \calP_k$. By the Parity Lemma~\ref{lem:parityLemma} one has $N_{\OPT\triangle (P\cup\calP_{k-1})}(u_2) = \{u_1, u_4\}$. Then $N_{\OPT\triangle(P\cup\calP_{k-1}\cup\{u_2u_3\})}(u_2) = \{u_1, u_3, u_4\}$. Therefore the only possible triangles in $\OPT\triangle(P\cup\calP_{k-1}\cup\{u_2u_3\})$ are $u_1u_2u_3$ and $u_2u_3u_4$. 
    
    By Claim~\ref{claim:blueCase3u3u4}, $u_3u_4\in\APX\cap P_k$, so $u_3u_4\in\APX\setminus(\OPT\cup\calP_{k-1})$. Recall that $u_1u_3\in\OPT\cap P_k$. By assumption of this case, $u_1u_3\in P$ and $u_3u_4\notin P$, so $u_1u_3, u_3u_4\notin \OPT\triangle(P\cup\calP_{k-1}\cup\{u_2u_3\})$, a contradiction. The claim follows. 

    \vspace{2.5mm}\noindent\textbf{(3.2.2) $\mathbf{u_3u_4}$ is the last edge of $\mathbf{c_i}$.} This case is illustrated in Figure~\ref{fig:redCase322}.
    \begin{figure}[ht]
        \centering
        \includegraphics{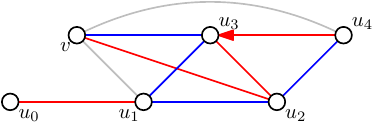}
        \caption{Proof of Lemma~\ref{lem:extendingTrailBlue}, Case \textbf{(3.2.2)}. Here the arrows are used for edges of $P$ instead of edges of $P_k$.}
        \label{fig:redCase322}
    \end{figure}
    \fab{By Remark~\ref{rem:notAugmenting} $P_1,\ldots,P_{k-1},P$, where $P=c_1\circ\dots\circ c_i$}, satisfy Invariant~\ref{inv:main}.     
    \begin{Claim}\label{claim:blueCase3u3v}
        There exists an edge $u_3v\in \OPT\setminus(\APX\cup \calP_k)$.
    \end{Claim}
    \begin{proof}
        By Claim~\ref{claim:blueCase3u2u3}, $u_2u_3\in\APX\setminus(\OPT\cup\calP_k)$. By assumption of Case~\textbf{(3.2)} and Claim~\ref{claim:blueCase3onlyPk}, $u_3u_4\in\APX\cap P_k$. Since $N_{\APX\setminus\OPT}(u_3)=\{u_2, u_4\}$ and $|N_{\OPT}(u_3)|\geq |N_{\APX}(u_3)|$ \fab{by the Degree Domination Lemma \ref{lem:blueMajority}}, there exist an edge $u_3v\in \OPT\setminus \APX, v\neq u_1$. By Claim~\ref{claim:blueCase3onlyPk}, one has $N_{P_k}(u_3)\fab{=}\{u_1, u_4\}$, so $u_3v\notin P_k$. If $u_3v\in P_j$ for some $j\in[1, k-1]$, since $u_2u_3, u_3u_4\in\APX\setminus\calP_{k-1}$, $N_{\APX\triangle P_j}(u_3)\supseteq \{u_2, u_4, v\}$, a contradiction by Corollary~\ref{cor:parityCorollary}. The claim follows.
    \end{proof}

    \begin{Claim}\label{claim:Case3contr}
        There is no chunk $c$ such that $P_1,\dots, P_{k-1}, P\circ c$ satisf\fab{y} Invariant~\ref{inv:main}, and the first edge of $c$ is $u_3v$. 
    \end{Claim}
    \begin{proof}
        Assume by contradiction that such $c$ exists. By Claim~\ref{claim:blueCase3u3v}, $u_3v\in\OPT\setminus(\APX\cup\calP_k)$. It must be \fab{the case} that there is no edge $u_4v\in \fab{\APX\cup \OPT}$. Indeed, otherwise $u_3v$ would share a triangle with the last edge of $c_i$, while the first edge of $c_{i+1}$ (\fab{namely,} $u_1u_3$) does not share a triangle with the last edge of $c_i$ (by Claim~\ref{claim:blueCase3u1u4}). This would be a \fab{violation of} Invariant~\ref{inv:tieBreak:destroyTriangle}. Recall that $u_1u_2\in\OPT\setminus(\APX\cup \calP_k)$. By Claim~\ref{claim:blueCase3u2u4}, $u_2u_4\in\OPT\setminus(\APX\cup \calP_k)$. Thus, there is no edge \fab{$vu_2$} in $\OPT\setminus \APX$ because $N_{\OPT}(u_2) = \{u_1, u_4\}$. 
        
        By assumption, $u_1u_3\in\OPT\cap P_k, u_3u_4\in\APX\setminus\OPT$. One has $N_\OPT(u_3)=\{u_1, v\}$, $N_\APX(u_3) = \{u_2, u_4\}$, there is no edge $u_4v$ in $\APX\cup\OPT$ and no edge $u_2v$ in $\OPT$. Therefore, every triangle of $\APX\cup\OPT$ that contains $u_3v$ and \fab{an}other edge of $\OPT$ must also contain the edge $u_1u_3$. Since $u_1u_2u_3$ is a triangle in $\APX\cup\OPT$ with $2$ edges from $\OPT$, this means that $u_3v$ is in \fab{fewer} such triangles than $u_1u_3$. This is a \fab{violation of} Invariant~\ref{inv:tieBreak}.\ref{inv:tieBreak:avoidTriangleBlue}.
    \end{proof}

    We will next show that there exists a chunk $c$ as described in Claim \ref{claim:Case3contr}, hence obtaining a contradiction.
By Claim~\ref{claim:blueCase3u3v}, $u_3v\in \OPT\setminus(\APX\cup \calP_k)$, so $c := u_3v$ is a chunk. Observe that \fab{$P_1,\ldots,P_{k-1},P\circ c$} satisfy Invariant~\ref{inv:main}.\ref{prp:chunks},~\ref{inv:main}.\ref{prp:disjoint} and~\ref{inv:main}.\ref{prp:startingEnding}. Since $u_3v\in\OPT$, and $\OPT\triangle(P\cup\calP_k)$ is triangle-free, Invariant~\ref{inv:main}.\ref{prp:blueTriangles} ($\OPT$-triangle) \fab{alsp holds}. If there is no triangle in $\APX\triangle (P\cup\{u_3v\})$ then Invariant~\ref{inv:main}.\ref{prp:redTriangles} ($\APX$-triangle) is also satisfied. Therefore, we can next assume that there is a triangle in $\APX\triangle (P\cup\{u_3v\})$. 
    
    Since $\APX\triangle P$ contains no triangle, every triangle in $\APX\triangle (P\cup\{u_3v\})$ must contain the edge $u_3v$. One has $u_3v\in\OPT\setminus(\APX\cup\calP_k)$, and by Claim~\ref{claim:blueCase3u2u3}, $u_2u_3\in\APX\setminus(\OPT\cup\calP_k)$, so by the Parity Lemma~\ref{lem:parityLemma} one has $N_{\APX\triangle (P\cup\{u_3v\})}(u_3) = \{u_2, v\}$. Therefore, the only triangle in $\APX\triangle (P\cup\{u_3v\})$ is $u_2u_3v$. 

    \begin{Claim}\label{claim:blueCase3u2v}
        $u_2v\in\APX\setminus(\OPT\cup\calP_k)$
    \end{Claim}
    \begin{proof}
        Recall that $u_1u_2\in\OPT\setminus(\APX\cup\calP_k)$. By Claim~\ref{claim:blueCase3u2u4}, $u_2u_4\in \OPT\setminus(\APX\cup\calP_k)$. Therefore, $u_2v\in\APX\setminus\OPT$, otherwise $N_{\OPT}(u_2) = \{u_1, u_4, v\}$, a contradiction to the fact that $\OPT$ is a $2$-matching. Moreover, by the Parity Lemma~\ref{lem:parityLemma}, $N_{\OPT\triangle\calP_k}(u_2) = \{u_1, u_4\}$, so $u_2v\notin\calP_k$. The claim follows.
    \end{proof}

    By Claims~\ref{claim:blueCase3u3v} and~\ref{claim:blueCase3u2v}, $u_3v\in\OPT\setminus(\APX\cup\calP_k)$ and $u_2v\in\APX\setminus(\OPT\cup\calP_k)$. Since $u_2u_3v$ is a triangle in $\APX\cup\OPT$, $c:=u_3v, vu_2$ is a chunk. We claim \fab{that $P_1,\ldots,P_{k-1},P\circ c$} satisfy Invariant~\ref{inv:main}. Observe that Invariants~\ref{inv:main}.\ref{prp:chunks},~\ref{inv:main}.\ref{prp:disjoint} and~\ref{inv:main}.\ref{prp:startingEnding} are satisfied. Recall that $u_2u_3v$ is the only triangle in $\APX\triangle(P\cup\{u_3v\})$. Since $u_2v\in\APX$, $\fab{u_2v\notin}\APX\triangle(P\cup\{u_2v, u_3v\})$, thus  Invariant~\ref{inv:main}.\ref{prp:redTriangles} ($\APX$-triangle) also holds. It remains to check  Invariant~\ref{inv:main}.\ref{prp:blueTriangles} ($\OPT$-triangle). By contradiction, assume that there is a triangle in $\OPT\triangle(P\cup\calP_{k-1}\cup\{u_2v, u_3v\})$.
    
    Since $\OPT\triangle(P\cup\calP_{k-1})$ contains no triangle and $u_3v\in\OPT$, every triangle in $\OPT\triangle(P\cup\calP_{k-1}\cup\{u_2v, u_3v\})$ must contain the edge $u_2v$. Recall that $u_1u_2\in\OPT\setminus(\APX\cup\calP_{k})$, and, by Claim~\ref{claim:blueCase3u2u4}, $u_2u_4\in\OPT\setminus(\APX\cup\calP_{k})$, implying $u_1u_2, u_2u_4\in\OPT\setminus(P\cup\calP_{k-1})$. By the Parity Lemma~\ref{lem:parityLemma}, $N_{\OPT\triangle(P\cup\calP_{k-1})}(u_2) = \{u_1, u_4\}$. Thus, one has $N_{\OPT\triangle(P\cup\calP_{k-1}\cup\{u_2v, u_3v\})}(u_2) = \{u_1, u_4, v\}$. Therefore, the only possible triangles in $\OPT\triangle(P\cup\calP_{k-1}\cup\{u_2v, u_3v\})$ are $u_1u_2v$ and $u_2u_4v$. 
        
    Since $u_1u_2\in\OPT\setminus(\APX\cup\calP_k)$ and, by assumption, $u_1u_3\in\OPT\cap P_k$, one has $u_1u_2, u_1u_3\in\OPT\setminus(P\cup\calP_{k-1})$. By the Parity Lemma~\ref{lem:parityLemma}, one has $N_{\OPT\triangle(P\cup\calP_{k-1})}(u_1) = \{u_2, u_3\}$. Thus, $N_{\OPT\triangle(P\cup\calP_{k-1}\cup\{u_2v, u_3v\})}(u_1) = \{u_2, u_3\}$, so $u_1v\notin\OPT\triangle(P\cup\calP_{k-1}\cup\{u_2v, u_3v\})$. Therefore $u_2u_4v$ must be a triangle in $\OPT\triangle(P\cup\calP_{k-1}\cup\{u_2v, u_3v\})$. 
    
    By assumptions of Case~\textbf{(3.2.2)}, $u_3u_4\in\APX\cap P$, and, by Claim~\ref{claim:blueCase3u2u4}, $u_2u_4\in\OPT\setminus(\APX\cup\calP_{k})$. By the Parity Lemma~\ref{lem:parityLemma}, $N_{\OPT\triangle(P\cup\{u_2v, u_3v\}\cup\calP_{k-1})}(u_4) = \{u_2, u_3\}$. Therefore, one has that $u_4v\notin\OPT\triangle(P\cup\{u_3v, vu_2\}\cup\calP_{k-1})$, a contradiction. 
\end{proof}

\begin{lemma}\label{lem:extendingTrailRed}
    Let $P_1,\dots, P_k$ be the current collection of alternating trails \fab{(satisfying Invariants \ref{inv:main} and \ref{inv:tieBreak})}, where $P_k$ is not augmenting or or the last node of $P_k$ is not deficient w.r.t. $\calP_{k-1}$. Suppose that the last edge \fab{$u_0u_1$} of $P_k$ is in $\OPT\setminus \APX$. Then there exist a chunk $c$ such that $P_1,\dots,P_k\circ c$ satisfy Invariants~\ref{inv:main} \fab{and \ref{inv:tieBreak}}.
\end{lemma}
\begin{proof}
\mi{It is sufficient to show that there exists one such chunk $c$ so that Invariant \ref{inv:main} holds. Invariant \ref{inv:tieBreak} can then be trivially enforced by breaking ties properly.}

By Lemma~\ref{lem:continueTrail}, there exist an edge $u_1u_2\in \APX\setminus (\OPT\cup\calP_k)$. Notice that $u_1u_2$ is a chunk. If $c := u_1u_2$ satisfies the claim we are done, hence assume this is not the case. Observe that $P_1,\dots, P_k\circ c$ satisfy Invariant~\ref{inv:main}.\ref{prp:chunks},~\ref{inv:main}.\ref{prp:disjoint} and~\ref{inv:main}.\ref{prp:startingEnding}. Since $u_1u_2\in \APX\setminus \OPT$ and $\APX\triangle P_k$ is triangle-free, Invariant~\ref{inv:main}.\ref{prp:redTriangles} ($\APX$-triangle) is also satisfied. Thus Invariant~\ref{inv:main}.\ref{prp:blueTriangles} ($\OPT$-triangle) must be violated. In more detail, $\OPT\triangle (\calP_k\cup\{u_1u_2\}) = (\OPT\triangle \calP_k)\cup \{u_1u_2\}$ must contain a triangle (and $u_1u_2$ must belong to that triangle since \fab{$\OPT\triangle \calP_k$} is triangle-free).\fabr{It was $\APX\triangle P_k$, check} By the Unique Triangle Corollary~\ref{cor:uniqueTriangle}, there is only one such triangle: let $u_1u_2u_3$ be the only triangle in $\OPT\triangle (\calP_k\cup\{u_1u_2\})$. Since $u_0u_1\in \OPT\cap P_k$, $u_0u_1\notin \OPT\triangle (\calP_k\cup \{u_1u_2\})$. Thus \fab{$u_0\neq u_3$, hence} $u_0, u_1, u_2, u_3$ are all distinct.

    We consider the following cases depending on the edge $u_1u_3$. Notice that, since $u_1u_3\in \OPT\triangle (\calP_k\cup \{u_1u_2\})$ then $u_1u_3$ is in either $\OPT\setminus \calP_k$ or in $\calP_k\setminus \OPT=\APX\cap \calP_k$. \mi{In the first case we further distinguish between the case $u_1u_3\in \OPT\setminus(\APX\cup\calP_k)$ and $u_1u_3\in \OPT\cap \APX$. In the second case we further distinguish between the case $u_1u_3\in\APX\cap\calP_{k-1}$ and $u_1u_3\in\APX\cap P_k$. The cases $u_1u_3\in \OPT\cap \APX$ and $u_1u_3\in\APX\cap\calP_{k-1}$ can be dealt with together (in Case~\textbf{(2)}), leading to the next $3$ cases.}\fabr{Case 2 requires an explanation}\mig{Added explaination}
        
    \vspace{2.5mm}\noindent\textbf{(1) $\mathbf{u_1u_3\in \OPT\setminus (\APX\cup\calP_k)}$.} We illustrate this case in Figure~\ref{fig:blueCase1}.
    \begin{figure}[ht]
        \centering
        \includegraphics{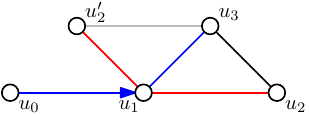}
        \caption{Proof of Lemma~\ref{lem:extendingTrailRed}, Case \textbf{(1)}.}
        \label{fig:blueCase1}
    \end{figure}
    \begin{Claim}\label{claim:redCase1u1u2'}
        There exists an edge $u_1u_2'\in\APX\setminus(\OPT\cup \calP_k), u_2'\neq u_2$.
    \end{Claim}
    \begin{proof}
        \fab{One has} $u_0u_1\in\OPT\cap P_k$, so $u_0u_1\in\OPT\setminus(\APX\cup\calP_{k-1})$. Also, $u_1u_2\in\APX\setminus(\OPT\cup\calP_k)$ and $u_1u_3\in \OPT\setminus(\APX\cup\calP_k)$. Therefore, if $N(u_1)=\{u_0, u_2, u_3\}$, then $u_1$ is is deficient w.r.t. $\calP_{k-1}$ and $P_k$ is augmenting, a contradiction to the assumptions of this lemma. Thus, there must exist an edge $u_1u_2'\in\APX\setminus\OPT, u_2'\neq u_2$. 
        
        By contradiction, assume that $u_1u_2'\in\calP_k$. Then, since $u_1u_2u_3$ is a triangle in $\OPT\triangle(\calP_k\cup\{u_1u_2\})$ and $u_1u_2'\in\APX\setminus\OPT$, $N_{\OPT\triangle(\calP_k\cup\{u_1u_2\})}(u_1) \supseteq\{u_2, u_2', u_3\}$, a contradiction to the Parity Lemma~\ref{lem:parityLemma}. The claim follows.
    \end{proof}

    By Claim~\ref{claim:redCase1u1u2'}, $c:=u_1u_2'$ is a chunk. We claim that $P_1,\dots, P_k\circ c$ satisfy Invariant~\ref{inv:main}. Observe that Invariants~\ref{inv:main}.\ref{prp:chunks},~\ref{inv:main}.\ref{prp:disjoint} and~\ref{inv:main}.\ref{prp:startingEnding} are satisfied. Since $u_1u_2'\in \APX\setminus \OPT$ and $\APX\triangle P_k$ is triangle-free, Invariant~\ref{inv:main}.\ref{prp:redTriangles} ($\APX$-triangle) is also satisfied. We next prove that Invariant~\ref{inv:main}.\ref{prp:blueTriangles} ($\OPT$-triangle) holds. By contradiction, assume that there is a triangle in $\OPT\triangle(\calP_k\cup\{u_1u_2'\})$.
    
    Since $\OPT\triangle \calP_k$ contains no triangle, every triangle in $\OPT\triangle (\calP_k\cup \{u_1u_2'\})$ must contain the edge $u_1u_2'$. By the Unique Triangle Corollary~\ref{cor:uniqueTriangle}, there is at most one such triangle. Since $u_1u_3\in\OPT\setminus(\APX\cup\calP_k)$ and, by Claim~\ref{claim:redCase1u1u2'}, $u_1u_2'\in\APX\setminus(\OPT\cup\calP_k)$, by the Parity Lemma~\ref{lem:parityLemma} one has $N_{\OPT\triangle(\calP_k\cup\{u_1u_2'\})}\fab{(u_1)} = \{\fab{u_2',u_3}\}$. Thus, the only triangle in $\OPT\triangle(\calP_k\cup\{u_1u_2'\})$ is $u_1u_2'u_3$. 
    
    Since $u_1u_2u_3$ is a triangle in $\OPT\triangle (\calP_k\cup\{u_1u_2\})$, by the Parity Lemma~\ref{lem:parityLemma}, one has $N_{\OPT\triangle (\calP_k\cup\{u_1u_2\})}(u_3) = \{u_1, u_2\}$. Observing that $N_{\OPT\triangle (\calP_k\cup\{u_1u_2'\})}(u_3) = N_{\OPT\triangle (\calP_k\cup\{u_1u_2\})}(u_3)$, we infer that $u_2'u_3\notin \OPT\triangle (\calP_k\cup\{u_1u_2'\})$, a contradiction \fab{(since $u_1u_2'u_3$ is a triangle in $\OPT\triangle(\calP_k\cup\{u_1u_2'\})$)}. 

    \vspace{2.5mm}\noindent\textbf{(2) $\mathbf{u_1u_3\in \APX\cap (\OPT\cup\calP_{k-1})}$.} We illustrate this case in Figure~\ref{fig:blueCase2}.
    \begin{figure}[ht]
        \centering
        \includegraphics{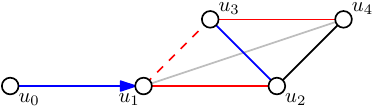}
        \caption{Proof of Lemma~\ref{lem:extendingTrailRed}, Case \textbf{(2)}. We illustrate the case when $u_1u_3\in \APX\cap\calP_{k-1}$, the other case (when $u_1u_3\in\APX\cap\OPT$) being identical.}
        \label{fig:blueCase2}
    \end{figure}

    \begin{Claim}\label{claim:redCase2u2u3}
        $u_2u_3\in\OPT\setminus(\APX\cup\calP_k)$.
    \end{Claim}
    \begin{proof}
        Recall that $u_1u_2\in\APX\setminus(\OPT\cup\calP_k)$ and, by Case~\textbf{(2)} assumption $u_1u_3\in\APX\cap(\OPT\cup\calP_{k-1})$, so $u_1u_3\in \APX\triangle P_k$. Notice that $u_2u_3\in\OPT\setminus\APX$, otherwise $u_1u_2u_3$ would be a triangle in $\APX$. Recall that $u_1u_2u_3$ is a triangle in $\OPT\triangle(\calP_k\cup\{u_1u_2\})$. Thus one has $u_2u_3\notin\calP_k$, otherwise $u_2u_3\notin\OPT\triangle(\calP_k\cup\{u_1u_2\})$, a contradiction. The claim follows.
    \end{proof}

Recall that $u_1u_2\in\APX\setminus(\OPT\cup\calP_k)$.\fabr{I removed 2 cases for uniformity with previous lemma (I changed my mind Miguel)} By Claim~\ref{claim:redCase2u2u3}, $u_2u_3\in\OPT\setminus(\APX\cup\calP_k)$. Also, $u_1u_2u_3$ is a triangle in $\APX\cup\OPT$, so $c:= u_1u_2, u_2u_3$ is a chunk. 

If there is no triangle in $\APX\triangle(P_k\cup\{u_1u_2, u_2u_3\})$, then $c$ satisfies the claim of the lemma. Indeed, $P_1,\dots, P_k\circ c$ satisf\fab{y} Invariants~\ref{inv:main}.\ref{prp:chunks},~\ref{inv:main}.\ref{prp:disjoint} and~\ref{inv:main}.\ref{prp:startingEnding}. Invariant~\ref{inv:main}.\ref{prp:redTriangles} ($\APX$-triangle) is also trivially satisfied. Recall that $u_1u_2u_3$ is the only triangle in $\OPT\triangle (\calP_k\cup\{u_1u_2\})$. Since $u_2u_3\in\OPT$, $\OPT\triangle (\calP_k\cup\{u_1u_2, u_2u_3\})$ contains no triangle and Invariant~\ref{inv:main}.\ref{prp:blueTriangles} ($\OPT$-triangle) holds. 

Hence we next assume that there is a triangle in $\APX\triangle(P_k\cup\{u_1u_2, u_2u_3\})$. By Invariant~\ref{inv:main}.\ref{prp:redTriangles} ($\APX$-triangle), $\APX\triangle P_k$ contains no triangles. Since $u_1u_2\in\APX$, every triangle in $\APX\triangle(P_k\cup\{u_1u_2, u_2u_3\})$ must contain the edge $u_2u_3$. By the Unique Triangle Corollary~\ref{cor:uniqueTriangle}, there is only one such triangle: let $u_2u_3u_4$ be the only triangle in $\APX\triangle(P_k\cup\{u_1u_2, u_2u_3\})$. Since $u_1u_2\in \APX$, $u_1u_2\notin \APX\triangle(P_k\cup\{u_1u_2, u_2u_3\})$. So $u_4\neq u_1$, and thus $u_1, u_2, u_3, u_4$ are all distinct.
      
    \begin{Claim}\label{claim:redCase2u3u4}
        $u_3u_4\in\APX\setminus(\OPT\cup\calP_k)$
    \end{Claim}
    \begin{proof}
        By Claim~\ref{claim:redCase2u2u3}, $u_2u_3\in\OPT\setminus(\APX\cup\calP_k)$. By contradiction, assume that $u_3u_4\in\OPT$. It must be the case that $u_1u_3\notin\OPT$ because otherwise $N_\OPT(u_3) \fab{\supseteq} \{u_1, u_2, u_4\}$, a contradiction \fab{since $\OPT$ is a 2-matching.} Thus by case \fab{\textbf{(2)}} assumption $u_1u_3\in\APX\cap\calP_{k-1}$. Since $u_2u_3u_4$ is a triangle in $\APX\triangle(P_k\cup\{u_1u_2, u_2u_3\})$ and $u_3u_4\in\OPT$, either $u_3u_4\in P_k$ or $u_3u_4\in\APX\cap\OPT$. In both cases one has $u_3u_4\notin\calP_{k-1}$. Therefore, $N_{\OPT\triangle \calP_{k-1}}(u_3) \fab{\supseteq} \{u_1, u_2, u_4\}$, a contradiction by Corollary~\ref{cor:parityCorollary}. Thus $u_3u_4\in\APX\setminus\OPT$. 
        
        Notice that by the assumption of Case~\textbf{(2)}, and since $u_2u_3\in\OPT\setminus\calP_k$, $u_1u_3, u_2u_3\in \OPT\triangle\calP_k$.\fabr{I think that $u_1u_3\in \OPT\triangle\calP_k$ needs some explanation. \af{agree!}} By the Parity Lemma~\ref{lem:parityLemma}, $N_{\OPT\triangle\calP_k}(u_3)=\{u_1, u_2\}$. Therefore, since $u_3u_4\in\APX\setminus\OPT$, it must be that $u_3u_4\notin\calP_k$. The claim follows.
    \end{proof}
    
    Recall that $u_1u_2\in\APX\setminus(\OPT\cup\calP_k)$. By Claims~\ref{claim:redCase2u2u3} and~\ref{claim:redCase2u3u4}, $u_2u_3\in\OPT\setminus(\APX\cup\calP_k)$ and $u_3u_4\in\APX\setminus(\OPT\cup\calP_k)$. Furthermore, $u_1u_2u_3$ and $u_2u_3u_4$ are triangles in $\APX\cup\OPT$, so $c := u_1u_2, u_2u_3, u_3u_4$ is a chunk. We claim that $P_1,\dots, P_k\circ c$ satisfies Invariant~\ref{inv:main}. 
        
    Observe that Invariants~\ref{inv:main}.\ref{prp:chunks},~\ref{inv:main}.\ref{prp:disjoint} and~\ref{inv:main}.\ref{prp:startingEnding} are satisfied. Recall that $u_2u_3u_4$ is the only triangle in $\APX\triangle(P_k\cup\{u_1u_2, u_2u_3\})$. Since $u_3u_4\in\APX$, $\APX\triangle(P_k\cup \af{c})$ contains no triangle, so Invariant~\ref{inv:main}.\ref{prp:redTriangles} ($\APX$-triangle) is also satisfied. It is only left to prove that Invariant~\ref{inv:main}.\ref{prp:blueTriangles} ($\OPT$-triangle) is also satisfied. By contradiction, assume that there is a triangle in $\OPT\triangle (\calP_k\cup\af{c})$. 
    
    Recall that $u_1u_2u_3$ is the only triangle in $\OPT\triangle (\calP_k\cup\{u_1u_2\})$, and, since $u_2u_3\in\OPT$, $u_2u_3\notin\OPT\triangle (\calP_k\cup\af{c})$. Thus, every triangle in $\OPT\triangle (\calP_k\cup\af{c})$ must contain the edge $u_3u_4$.

    We have $u_3u_4\in \APX\setminus\OPT$ and, by assumption of Case~\textbf{(2)}, $u_1u_3\in \APX\cap (\OPT\cup\calP_{k-1})$, so by the Parity Lemma~\ref{lem:parityLemma}, $N_{\OPT\triangle (\calP_k\cup\af{c})}(u_3) = \{u_1, u_4\}$. Thus, the only possible triangle in $\OPT\triangle (\calP_k\cup c)$ is $u_1u_3u_4$. Since $u_1u_2\in\APX\setminus\OPT$, by the Parity Lemma~\ref{lem:parityLemma}, $N_{\OPT\triangle (\calP_k\cup c)}(u_1) = \{u_2, u_3\}$, implying $u_1u_4\notin \OPT\triangle (\calP_k\cup c)$, a contradiction.

    \vspace{2.5mm}\noindent\textbf{(3) $\mathbf{u_1u_3\in \APX\cap P_k}$.}

    \begin{Claim}\label{claim:redCase3u2u3}
        $u_2u_3\in \OPT\setminus(\APX\cup \calP_k)$.
    \end{Claim}
    \begin{proof}
        Recall $u_1u_2\in\APX\setminus(\OPT\cup\calP_k)$. Since $u_1u_3\in \fab{\APX\cap P_k\subseteq}\APX\setminus \OPT$, one has $u_2u_3\in \OPT\setminus \APX$ (otherwise $u_1u_2u_3$ would be a triangle in $\APX$). Since $u_1u_2u_3$ is a triangle in $\OPT\triangle (\calP_k\cup\{u_1u_2\})$, $u_2u_3\notin \calP_k$. The claim follows.
    \end{proof}

\fab{Since $u_1u_2\in \APX\setminus \OPT$, $u_2u_3\in \OPT\setminus \APX$ and $u_1u_2u_3$ is a triangle, $c:=u_1u_2,u_2u_3$ is a chunk.} If\fabr{Removed 2 subcases like in (2)} $\APX\triangle(P_k\cup\{u_1u_2, u_2u_3\})$ contains no triangle, then $c$ satisfies the claim of the lemma by the same argument as in case $\mathbf{(2)}$ (with Claim~\ref{claim:redCase3u2u3} in place of Claim~\ref{claim:redCase2u2u3}).

    


Thus we next assume that $\APX\triangle(P_k\cup\{u_1u_2, u_2u_3\})$ contains a triangle. By Invariant~\ref{inv:main}.\ref{prp:redTriangles} ($\APX$-triangle), $\APX\triangle P_k$ contains no triangles. Since $u_1u_2\in\APX$, every triangle in $\APX\triangle(P_k\cup\{u_1u_2, u_2u_3\})$ must contain the edge $u_2u_3$. By the Unique Triangle Corollary~\ref{cor:uniqueTriangle}, there is only one such triangle: let $u_2u_3u_4$ be the only triangle in $\APX\triangle(P_k\cup\{u_1u_2, u_2u_3\})$. Since $u_1u_2\in \APX$, $u_1u_2\notin \APX\triangle(P_k\cup\{u_1u_2, u_2u_3\})$. So $u_4\neq u_1$ and thus $u_1, u_2, u_3, u_4$ are all distinct. We distinguish two cases:
        
    \vspace{2.5mm}\noindent\textbf{(3.1) $\mathbf{u_3u_4\in \APX}$.} We illustrate this case in Figure~\ref{fig:blueCase31}.
    \begin{figure}[ht]
        \centering
        \includegraphics{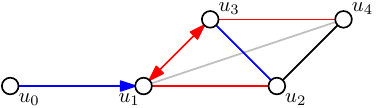}
        \caption{Proof of Lemma~\ref{lem:extendingTrailRed}, Case \textbf{(3.1)}.}
        \label{fig:blueCase31}
    \end{figure}
    \begin{Claim}\label{claim:redCase3u3u4}
        $u_3u_4\in \APX\setminus (\OPT\cup\calP_k)$.
    \end{Claim}
    \begin{proof}
        By contradiction, assume that $u_3u_4\in \OPT\cup\calP_k$. Then, since $u_3u_4\in\APX$, $u_3u_4\notin\OPT\cap\calP_k$ and thus  $u_3u_4\in\OPT\triangle\calP_k$. One has $u_1u_3\in\APX\cap\calP_k$ by the assumption of Case~\textbf{(3)}, and $ u_2u_3\in\OPT\setminus\calP_k$ by Claim~\ref{claim:redCase3u2u3}. Thus $N_{\OPT\triangle\calP_k}(u_3) \supseteq \{u_1, u_2, u_4\}$, a contradiction to the Parity Lemma~\ref{lem:parityLemma}. 
    \end{proof}    
    
    Recall that $u_1u_2\in\APX\setminus(\OPT\cup\calP_k)$, and by Claims~\ref{claim:redCase3u2u3} and~\ref{claim:redCase3u3u4}, $u_2u_3\in\OPT\setminus(\APX\cup\calP_k)$ and $u_3u_4\in\APX\setminus(\OPT\cup\calP_k)$. Furthermore, $u_1u_2u_3$ and $u_2u_3u_4$ are triangles in $\APX\cup\OPT$, so $c:=u_1u_2, u_2u_3, u_3u_4$ is a chunk. We claim that $P_1,\dots, P_k\circ c$ satisfy Invariant~\ref{inv:main}.
    
    Observe $P_1,\dots, P_k\circ c$ satisfy Invariant~\ref{inv:main}.\ref{prp:chunks},~\ref{inv:main}.\ref{prp:disjoint} and~\ref{inv:main}.\ref{prp:startingEnding}. Recall that $u_2u_3u_4$ is the only triangle in $\APX\triangle(P_k\cup\{u_1u_2, u_2u_3\})$. Since $u_3u_4\in\APX$, $\APX\triangle(P_k\cup c)$ contains no triangle, so Invariant~\ref{inv:main}.\ref{prp:redTriangles} ($\APX$-triangle) is also satisfied. We next prove that Invariant~\ref{inv:main}.\ref{prp:blueTriangles} holds. By contradiction, assume that there is a triangle in $\OPT\triangle (\calP_k\cup c)$.

    Recall that $u_1u_2u_3$ is the only triangle in $\OPT\triangle (\calP_k\cup\{u_1u_2\})$. Since $u_2u_3\in \OPT$, $u_2u_3\notin\OPT\triangle (\calP_k\cup c)$. Thus, every triangle in $\OPT\triangle (\calP_k\cup c)$ must contain the edge $u_3u_4$. We have $u_3u_4\in\APX\setminus\OPT$ and, by the assumption of Case~\textbf{(3)}, $u_1u_3\in\APX\cap P_k$, so by the Parity Lemma~\ref{lem:parityLemma}, $N_{\OPT\triangle (\calP_k\cup c)}(u_3) = \{u_1, u_4\}$. Thus, the only possible triangle in $\OPT\triangle (\calP_k\cup c)$ is $u_1u_3u_4$. Since $u_1u_2\in \APX\setminus\OPT$, by the Parity Lemma~\ref{lem:parityLemma}, one has $N_{\OPT\triangle (\calP_k\cup c)}(u_1) = \{u_2, u_3\}$, implying $u_1u_4\notin \OPT\triangle (\calP_k\cup c)$, a contradiction. 

    \vspace{2.5mm}\noindent\textbf{(3.2) $\mathbf{u_3u_4\in \OPT\setminus \APX}$.} We will show that this case is not possible. We illustrate this case in Figure~\ref{fig:blueCase32}.
    \begin{figure}[ht]
        \centering
        \includegraphics{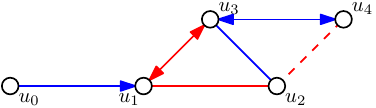}
        \caption{Proof of Lemma~\ref{lem:extendingTrailRed}, Case \textbf{(3.2)}.}
        \label{fig:blueCase32}
    \end{figure}
    
    \begin{Claim}\label{claim:redCase3u2u4}
        $u_2u_4\in \APX\setminus (\OPT\cup P_k)$.
    \end{Claim}
    \begin{proof}
        By Claim~\ref{claim:redCase3u2u3} and by the assumption of Case~\textbf{(3.2)}, $u_2u_3, u_3u_4\in\OPT$, so $u_2u_4\in \APX\setminus\OPT$, otherwise $u_2u_3u_4$ would be a triangle in $\OPT$. Furthermore, since $u_2u_3u_4$ is a triangle in $\APX\triangle(P_k\cup\{u_1u_2, u_2u_3\})$, $u_2u_4\notin P_k$.
    \end{proof}

    \begin{Claim}\label{claim:redCase3onlyPk}
        $N_{P_k}(u_3) = \{u_1, u_4\}$
    \end{Claim}
    \begin{proof}
        Since $u_2u_3u_4$ is a triangle in $\APX\triangle(P_k\cup\{u_1u_2, u_2u_3\})$, and $u_3u_4\in\OPT\setminus\APX$ by the assumption of Case~\textbf{(3.2)}, it must be that $u_3u_4\in P_k$. Thus, $N_{P_k}(u_3)\supseteq \{u_1, u_4\}$. By Claim~\ref{claim:redCase3u2u3}, $u_2u_3\in\OPT\setminus(\APX\cup\calP_k)$, so $u_2u_3\notin P_k$. If $u_3$ has another neighbor $v\notin\{u_1, u_2, u_4\}$ \fab{w.r.t. $\APX\cup \OPT$}, it must be that $u_3v\in\APX\setminus\OPT$, because $N_\OPT(u_3)=\{u_2, u_4\}$. Moreover, since $u_1u_3\in \APX\cap P_k$ by the assumption of Case~\textbf{(3)}, by the Parity Lemma~\ref{lem:parityLemma} one has $N_{\OPT\triangle \calP_k}(u_3) = \{u_1, u_2\}$. Thus, it must be that $u_3v\notin \calP_k$.
    \end{proof}

    \begin{Claim}
        $u_4\neq u_0$.
    \end{Claim}
    \begin{proof}
        By Claim~\ref{claim:redCase3u2u4}, $u_2u_4\in\APX\setminus(\OPT\cup P_k)$, and by Claim~\ref{claim:redCase3u3u4}, $u_3u_4\in\OPT\cap P_k$. By contradiction, assume that $u_0=u_4$. There is an edge $u_0u_3\in \OPT\cap P_k$ and an edge $u_0u_2\in\APX\setminus (\OPT\cup P_k)$. But then, since $u_0u_1\in\OPT\cap P_k$, $N_{\APX\triangle P_k}(u_0) \supseteq \{u_1, u_2, u_3\}$, a contradiction by the Parity Lemma~\ref{lem:parityLemma}.
        \end{proof}

    \begin{Claim}\label{claim:redCase3u1u4}
        There is no edge $u_1u_4$ in $\APX\cup\OPT$.
    \end{Claim}
    \begin{proof}
        Recall that $u_1u_2\in\APX\setminus(\OPT\cup\calP_k)$, and, by assumption of Case~\textbf{(3)}, $u_1u_3\in\APX\cap P_k$. Assume to get a contradiction that $u_1u_4\in \APX\cup\OPT$. Then since $N_\APX(u_1)=\{u_2, u_3\}$ it must be that $u_1u_4\in\OPT\setminus\APX$. 
        
        If $u_1u_4\in P_j$ for some $j\in [1, k-1]$ we reach a contradiction. Indeed, we have $u_1u_2, u_1u_3\in\APX\setminus P_j$, so $N_{\APX\triangle P_j}(u_1) \supseteq \{u_2, u_3, u_4\}$, a contradiction by Corollary~\ref{cor:parityCorollary}. By Claim~\ref{claim:redCase3u2u4}, $u_2u_4\in\APX\setminus(\OPT\cup P_k)$. By the assumption of Case~\textbf{(3.2)}, $u_3u_4\in\OPT\setminus\APX$, and by Claim~\ref{claim:redCase3onlyPk}, $u_3u_4\in\OPT\cap P_k$. If $u_1u_4\in P_k$ then, one has $N_{\APX\triangle P_k}(u_4)\supseteq \{u_1, u_2, u_3\}$, a contradiction by the Parity Lemma~\ref{lem:parityLemma}. Therefore, $u_1u_4\in\OPT\setminus(\APX\cup\calP_k)$.
        
        Since $u_1u_2\in\APX\setminus\OPT$ and $u_1u_3\in\APX\cap P_k$, one has $N_{\OPT\triangle (\calP_k\cup \{u_1u_2\})}(u_1) \supseteq \{u_2, u_3, u_4\}$, a contradiction by the Parity Lemma~\ref{lem:parityLemma}. The claim follows.
    \end{proof}

    By Invariant~\ref{inv:main}.\ref{prp:chunks}, $P_k = c_1\circ\dots\circ c_l$, where $c_j$ is a chunk for all $j\in[1, l]$. By Claim~\ref{claim:redCase3u1u4}, there is no edge $u_1u_4$ \fab{in} $\APX\cup\OPT$, so there is no triangle $u_1u_3u_4$ in $\APX\cup\OPT$. By Claim~\ref{claim:redCase3onlyPk}, $N_{P_k}(u_3)=\{u_1, u_4\}$. Therefore, by Definition~\ref{def:chunk} of chunk, the edges $u_1u_3$ and $u_3u_4$ must belong to consecutive chunks in $P_k$. Thus, we have $P_k = c_1\circ\dots\circ c_i\circ c_{i+1}\circ\dots c_l$, where $c_i$ is either a chunk whose last edge is $u_1u_3$ (and the first edge of $c_{i+1}$ is $u_3u_4$) or $c_i$ is a chunk whose last edge is $u_3u_4$ (and the first edge of $c_{i+1}$ is $u_1u_3$). We consider those two cases.
    
    \vspace{2.5mm}\noindent\textbf{(3.2.1) $\mathbf{u_1u_3}$ is the last edge of $\mathbf{c_i}$.} This case is illustrated in Figure~\ref{fig:blueCase321}.
    \begin{figure}[ht]
        \centering
        \includegraphics{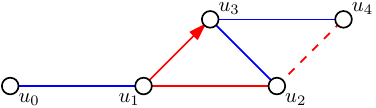}
        \caption{Proof of Lemma~\ref{lem:extendingTrailRed}, Case \textbf{(3.2.1)}. Here the arrows are used for edges of $P$ instead of edges of $P_k$.}
        \label{fig:blueCase321}
    \end{figure}
    By Remark~\ref{rem:notAugmenting}, $P_1,\dots,P_{k-1},P$, where $P=c_1\circ\dots\circ c_i$, must satisfy Invariant~\ref{inv:main}. By Claim~\ref{claim:redCase3u2u3}, $u_2u_3\in\OPT\setminus(\APX\cup\calP_k)$, so $c:=u_2u_3$ is a chunk. We claim that $P_1,\dots, P_{k-1}, P\circ c$ satisfy Invariant~\ref{inv:main}. That would lead to a contradiction, by Invariant~\ref{inv:tieBreak}.\ref{inv:tieBreak:destroyTriangle}. Indeed, $u_1u_2u_3$ is a triangle in $\APX\cup\OPT$, so the first edge of $c$ (\fab{namely,} $u_2u_3$) shares a triangle in $\APX\cup\OPT$ with the last edge of $c_i$ (\fab{namely,} $u_1u_3$). On the other hand, since by Claim~\ref{claim:redCase3u1u4}, $u_1u_4\notin\APX\cup\OPT$, the first edge of $c_{i+1}$ (\fab{namely,} $u_3u_4$), does not share a triangle in $\APX\cup\OPT$ with the last edge of $c_i$ (\fab{namely,} $u_1u_3$).

    Observe that $P_1,\dots, P_{k-1}, P\circ c$ satisfy  Invariant~\ref{inv:main}.\ref{prp:chunks},~\ref{inv:main}.\ref{prp:disjoint} and~\ref{inv:main}.\ref{prp:startingEnding}. Since $u_2u_3\in \OPT\setminus \APX$, and $\OPT\triangle(P\cup\calP_{k-1})$ is triangle-free, Invariant~\ref{inv:main}.\ref{prp:blueTriangles} ($\OPT$-triangle) is also satisfied. Thus, it remains to show that Invariant~\ref{inv:main}.\ref{prp:redTriangles} ($\APX$-triangle) holds. Assume by contradiction that $\APX\triangle(P\cup\{u_2u_3\})$ contains a triangle.

    Since $P_1,\dots,P_{k-1}, P$ satisfy Invariant~\ref{inv:main}.\ref{prp:redTriangles}, $\APX\triangle P$ contains no triangle. Therefore, every triangle in $\APX\triangle(P\cup\{u_2u_3\})$ must contain the edge $u_2u_3$. Recall that $u_1u_2\in\APX\setminus(\OPT\cup\calP_k)$. By Claim~\ref{claim:redCase3u2u4}, $u_2u_4\in\APX\setminus P_k$. By the Parity Lemma~\ref{lem:parityLemma} one has $N_{\APX\triangle P}(u_2) = \{u_1, u_4\}$. Then, $N_{\APX\triangle(P\cup\{u_2u_3\})}(u_2) = \{u_1, u_3, u_4\}$. Therefore the only possible triangles in $\APX\triangle(P\cup\{u_2u_3\})$ are $u_1u_2u_3$ and $u_2u_3u_4$. 

    By assumption of Cases~\textbf{(3)} and~\textbf{(3.2)}, $u_1u_3\in\APX\cap P_k$ and $u_3u_4\in\OPT\setminus\APX$. By assumption of Case~\textbf{(3.2.1)} $u_1u_3\in P$ and $u_3u_4\notin P$, so $u_1u_3, u_3u_4\notin \APX\triangle(P\cup\{u_2u_3\})$, a contradiction. 

    \vspace{2.5mm}\noindent\textbf{(3.2.2) $\mathbf{u_3u_4}$ is the last edge of $\mathbf{c_i}$.} This case is illustrated in Figure~\ref{fig:blueCase322}.
    \begin{figure}[ht]
        \centering
        \includegraphics{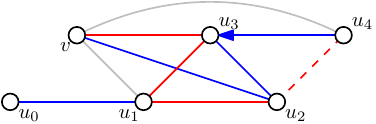}
        \caption{Proof of Lemma~\ref{lem:extendingTrailRed}, Case \textbf{(3.2.2)}. Here the arrows are used for edges of $P$ instead of edges of $P_k$.}
        \label{fig:blueCase322}
    \end{figure}
    By Remark~\ref{rem:notAugmenting}, $P_1,\dots,P_{k-1}, P$, where $P:=c_1\circ\dots\circ c_i$, satisfy Invariant~\ref{inv:main}.
    
    \begin{Claim}\label{claim:redCase3u3v}
        There exists an edge $u_3v\in\APX\setminus(\OPT\cup\calP_k)$
    \end{Claim}
    \begin{proof}
        By assumption of Cases~\textbf{(3)} and~\textbf{(3.2)}, $u_1u_3\in\APX\cap P_k$ and $u_3u_4\in\OPT\setminus\APX$. By Claim~\ref{claim:redCase3onlyPk}, $u_3u_4\in\OPT\cap P_k$. By Claim~\ref{claim:redCase3u2u3}, $u_2u_3\in \OPT\setminus(\APX\cup\calP_k)$. If there is no edge $u_3v\in\APX\setminus\OPT, v\neq u_1$, then $N(u_3) = \{u_1, u_2, u_4\}$ and thus $u_3$ is deficient. Since $u_1u_3, u_2u_3, u_3u_4\notin \calP_{k-1}$, $u_3$ is deficient w.r.t. $\calP_{k-1}$ and $P$ would be an augmenting trail, a contradiction to Remark~\ref{rem:notAugmenting}. Thus there is an edge $u_3v\in\APX\setminus\OPT$.
        
        Since $u_1u_3\in \APX\cap P_k$ and $u_2u_3\in\OPT\setminus\calP_k$, if $u_3v\in \calP_k$ then $N_{\OPT\triangle\calP_k}(u_3)\supseteq \{u_1, u_2, v\}$ a contradiction to the Parity Lemma~\ref{lem:parityLemma}, so $u_3v\notin\calP_k$. The claim follows.
    \end{proof}
    \begin{Claim}\label{claim:redCase3contr}
        There is no chunk $c$ such that $P_1,\dots, P_{k-1}, P\circ c$ satisf\fab{y} Invariant~\ref{inv:main}, and the first edge of $c$ is $u_3v$. 
    \end{Claim}
    \begin{proof}
        Assume by contradiction that there is a chunk $c$
        as in the claim. 
        Then, it must be the case that there is no edge $u_4v\in \APX\cup\OPT$. Indeed, otherwise the first edge of $c$ (namely, $u_3v$) would share a triangle with the last edge of $c_i$ (namely, $u_3u_4$), while, by Claim~\ref{claim:redCase3u1u4}, the first edge of $c_{i+1}$ (namely, $u_1u_3$) does not share a triangle with the last edge of $c_i$: this would be a violation of Invariant~\ref{inv:tieBreak}.\ref{inv:tieBreak:destroyTriangle}. Recall that $u_1u_2\in\APX\setminus(\OPT\cup\calP_k)$. By Claim~\ref{claim:redCase3u2u4}, $u_2u_4\in\APX\setminus(\OPT\cup P_k)$. Thus, there is no edge $vu_2\in\APX$ because $N_{\APX}(u_2) = \{u_1, u_4\}$.

        By assumption of Cases~\textbf{(3)} and~\textbf{(3.2)}, $u_1u_3\in\APX\cap P_k$ and $u_3u_4\in\OPT\setminus\APX$. By Claim~\ref{claim:redCase3u2u3}, $u_2u_3\in\OPT\setminus(\APX\cap\calP_k)$. One has $N_\APX(u_3)=\{u_1, v\}, N_\OPT(u_3) = \{u_2, u_4\}$, there is no edge $u_4v\in\APX\cup\OPT$ and no edge $u_2v\in\APX$. Therefore, every triangle that contains $u_3v$ and another edge of $\APX$ must also contain the edge $u_1u_3$. Since $u_1u_2u_3$ is a triangle with $2$ edges in $\OPT$, this implies that $u_3v$ is in fewer such triangles than $u_1u_3$. This is a violation of Invariant~\ref{inv:tieBreak}.\ref{inv:tieBreak:avoidTriangleRed}.
    \end{proof}

    We will next show that there exists a chunk $c$ as described in Claim~\ref{claim:redCase3contr}, hence obtaining a contradiction. By Claim~\ref{claim:redCase3u3v}, $u_3v\in\APX\setminus(\OPT\cup\calP_k)$, so $c := u_3v$ is a chunk. Observe that Invariants~\ref{inv:main}.\ref{prp:chunks},~\ref{inv:main}.\ref{prp:disjoint} and~\ref{inv:main}.\ref{prp:startingEnding} are satisfied. Since $u_3v\in\APX\setminus\OPT$ and $\APX\triangle P$ is triangle-free, Invariant~\ref{inv:main}.\ref{prp:redTriangles} ($\APX$-triangle) is also satisfied. If there is no triangle in $\OPT\triangle (P\cup\calP_{k-1}\cup\{u_3v\})$ then Invariant~\ref{inv:main}.\ref{prp:blueTriangles} ($\OPT$-triangle) holds. Therefore, we can assume by contradiction that there is a triangle in $\OPT\triangle (P\cup\calP_{k-1}\cup\{u_3v\})$. 

    Since $\OPT\triangle(P\cup\calP_{k-1})$ contains no triangle, every triangle in $\OPT\triangle (P\cup\calP_{k-1}\cup\{u_3v\})$ must contain the edge $u_3v$. Recall that $u_3v\in\APX\setminus(\OPT\cup\calP_k)$, and, by Claim~\ref{claim:redCase3u2u3}, $u_2u_3\in\OPT\setminus(\APX\cup\calP_k)$. Thus by the Parity Lemma~\ref{lem:parityLemma} one has $N_{\OPT\triangle(P\cup\calP_{k-1}\cup\{u_3v\})}(u_3)=\{u_2, v\}$. Therefore, the only triangle in $\OPT\triangle(P\cup\calP_{k-1}\cup\{u_3v\})$ is $u_2u_3v$.

    \begin{Claim}\label{claim:redCase3u2v}
        $u_2v\in\OPT\setminus(\APX\cup\calP_{k})$
    \end{Claim}
    \begin{proof}
        Recall that $u_1u_2\in\APX\setminus(\OPT\cup\calP_k)$. By Claim~\ref{claim:redCase3u2u4}, $u_2u_4\in\APX\setminus(\OPT\cup P_k)$. Therefore, $u_2v\in\OPT\setminus\APX$, otherwise $N_{\APX}(u_2)=\{u_1, u_4, v\}$, a contradiction to the fact that $\APX$ is a $2$-matching. Since $u_2u_3v$ is a triangle in $\OPT\triangle(P\cup\calP_{k-1}\cup\{u_3v\})$, $u_2v\notin\calP_{k-1}$. Moreover, by the Parity Lemma~\ref{lem:parityLemma}, $N_{\APX\triangle P_k}(u_2)=\{u_1, u_4\}$, so $u_2v\notin P_k$. The claim follows.  
    \end{proof}

    By Claims~\ref{claim:redCase3u3v} and~\ref{claim:redCase3u2v}, $u_3v\in\APX\setminus(\OPT\cup\calP_k)$ and $u_2v\in\OPT\setminus(\APX\cup\calP_k)$. Since $u_2u_3v$ is a triangle in $\APX\cup\OPT$, $c:=u_3v, vu_2$ is a chunk. We claim that $P_1,\dots, P_{k-1}, P\circ c$ satisfy Invariant~\ref{inv:main}, implying a contradiction by Claim~\ref{claim:redCase3contr}. Observe that Invariants~\ref{inv:main}.\ref{prp:chunks},~\ref{inv:main}.\ref{prp:disjoint} and~\ref{inv:main}.\ref{prp:startingEnding} are satisfied. Recall that $u_2u_3v$ is the only triangle in $\OPT\triangle(P\cup\calP_{k-1}\cup\{u_3v\})$. Since $u_2v\in\OPT$, $u_2v\notin\OPT\triangle(P\cup\calP_{k-1}\cup\{u_3v, u_2v\})$. Thus Invariant~\ref{inv:main}.\ref{prp:blueTriangles} ($\OPT$-triangle) is also satisfied. We next show that Invariant~\ref{inv:main}.\ref{prp:redTriangles} ($\APX$-triangle) holds. Assume by contradiction that there is a triangle in $\APX\triangle(P\cup\{u_3v, u_2v\})$.

    Since $\APX\triangle P$ contains no triangle and $u_3v\in\APX$, every triangle in $\APX\triangle(P\cup\{u_3v, u_2v\})$ must contain the edge $u_2v$. Recall that $u_1u_2\in\APX\setminus(\OPT\cup\calP_k)$, and, by Claim~\ref{claim:redCase3u2u4}, $u_2u_4\in\APX\setminus P_k$. Since $P\subseteq P_k$, by The Parity Lemma~\ref{lem:parityLemma}, $N_{\APX\triangle P}(u_2) =\{u_1, u_4\}$. Thus, one has $N_{\APX\triangle(P\cup\{u_3v, u_2v\})}(u_2) = \{u_1, u_4, v\}$. Therefore, the only possible triangles in $\APX\triangle(P\cup\{u_3v, u_2v\})$ are $u_1u_2v$ and $u_2u_4v$.

    Recall that $u_1u_2\in\APX\setminus(\OPT\cup\calP_k)$. By assumption of Case~\textbf{(3)}, $u_1u_3\in\APX\cup P_k$. By assumption of Case~\textbf{(3.2.2)}, $u_1u_3\in\APX\setminus P$. Therefore $u_1u_2, u_1u_3\in \APX\setminus P$. By the Parity Lemma~\ref{lem:parityLemma}, one has $N_{\APX\triangle P}(u_1) = \{u_2, u_3\}$. Thus, $N_{\APX\triangle(P\cup\{u_2v, u_3v\})}(u_1) = \{u_2, u_3\}$, so $u_1v\notin\APX\triangle(P\cup\{u_2v, u_3v\})$. Therefore, it must be that $u_2u_4v$ is a triangle in $\APX\triangle(P\cup\{u_3v, u_2v\})$.
    
    By assumption of Case~\textbf{(3.2)}, $u_3u_4\in\OPT\setminus\APX$. By assumption of Case~\textbf{(3.2.2)}, $u_3u_4\in\OPT\cap P$. By Claim~\ref{claim:redCase3u2u4}, $u_2u_4\in\APX\setminus P$. By the Parity Lemma~\ref{lem:parityLemma}, $N_{\APX\triangle(P\cup\{u_2v, u_3v\})}(u_4) = \{u_2, u_3\}$. Therefore, one has that $u_4v\notin\APX\triangle(P\cup\{u_2v, u_3v\})$, a contradiction (since $u_2u_4v$ is a triangle in $\APX\triangle(P\cup\{u_3v, u_2v\})$). 
\end{proof}

\end{document}